%% file: main.tex
\documentclass[11pt]{article}      
\usepackage[margin=1in,bottom=1in]{geometry}  
\usepackage{libertine}
\usepackage{amsthm}
\usepackage{amsmath}

\usepackage{amssymb}
\usepackage{stmaryrd}
\usepackage{mathrsfs}
\usepackage{mathtools}
\usepackage{thm-restate}
\usepackage[hidelinks]{hyperref}
\usepackage[capitalize]{cleveref}
\usepackage{url}
\usepackage{todonotes}
\usepackage{mathrsfs}
\usepackage{soul}
\usepackage{relsize}
\usepackage{enumerate}
\usepackage[all,defaultlines=3]{nowidow}
\usepackage[mathlines]{lineno}
\usepackage{multicol}
\usepackage{enumitem}
\usepackage{xspace}
\usepackage{textpos}
\setlist[itemize]{topsep=4pt,itemsep=3pt,parsep=0pt}
\setlist[enumerate]{topsep=4pt,itemsep=3pt,parsep=0pt}

\newenvironment{claimproof}[1][\proofname]{%
  \begin{proof}[#1]%
}{%
  \end{proof}%
}

\newcommand{\funding}{The work of Mi. P. was supported by the project BOBR that is funded from the European Research Council (ERC) under the European Union’s Horizon 2020 research and innovation programme with grant agreement No. 948057.}

\title{Parameterized and approximation algorithms for coverings points with segments in the plane\footnote{\funding}} 

\date{\today}
\author{
  Katarzyna Kowalska \\
  \small{University of Warsaw}\\
  \small{kasia.kowalska36@gmail.com}
  \and
  Michał Pilipczuk \\
  \small{University of Warsaw} \\
  \small{michal.pilipczuk@mimuw.edu.pl}
}

\crefname{claim}{Claim}{Claims}
\crefname{figure}{Figure}{Figures}
\newtheorem{theorem}{Theorem}

\newtheorem{lemma}[theorem]{Lemma}

\newtheorem{claim}{Claim}

\theoremstyle{definition}
\newtheorem{definition}{Definition}
\theoremstyle{plain}

\theoremstyle{definition}

\newtheorem*{example*}{Example}

\newcommand{\points}{\mathcal{U}}
\newcommand{\sets}{\mathcal{F}}
\newcommand{\sol}{\mathcal{S}}

\newcommand{\Oh}{\mathcal{O}}

\usepackage{todonotes}
\usepackage{xspace}
\usepackage{amsfonts}
\usepackage{amsmath}
\usepackage{graphicx}
\usepackage{xcolor}
\usepackage[nospace, noadjust]{cite}
\usepackage{lineno}
\usepackage{amsthm}
\usepackage{thmtools}
\usepackage{mathtools}  
\usepackage{makecell}
\usepackage{tikz}
\usetikzlibrary{calc,math}
\usepackage{hyperref}
\hypersetup{
    colorlinks,
    citecolor=black,
    filecolor=black,
    linkcolor=black,
    urlcolor=black
}

\mathtoolsset{showonlyrefs}

\renewcommand{\leq}{\leqslant}
\renewcommand{\geq}{\geqslant}
\renewcommand{\le}{\leqslant}
\renewcommand{\ge}{\geqslant}
\renewcommand{\setminus}{-}
\newcommand{\opt}{\mathsf{opt}}

\begin{document}
\maketitle

\begin{textblock}{20}(-1.75, 8.8)
\includegraphics[width=40px]{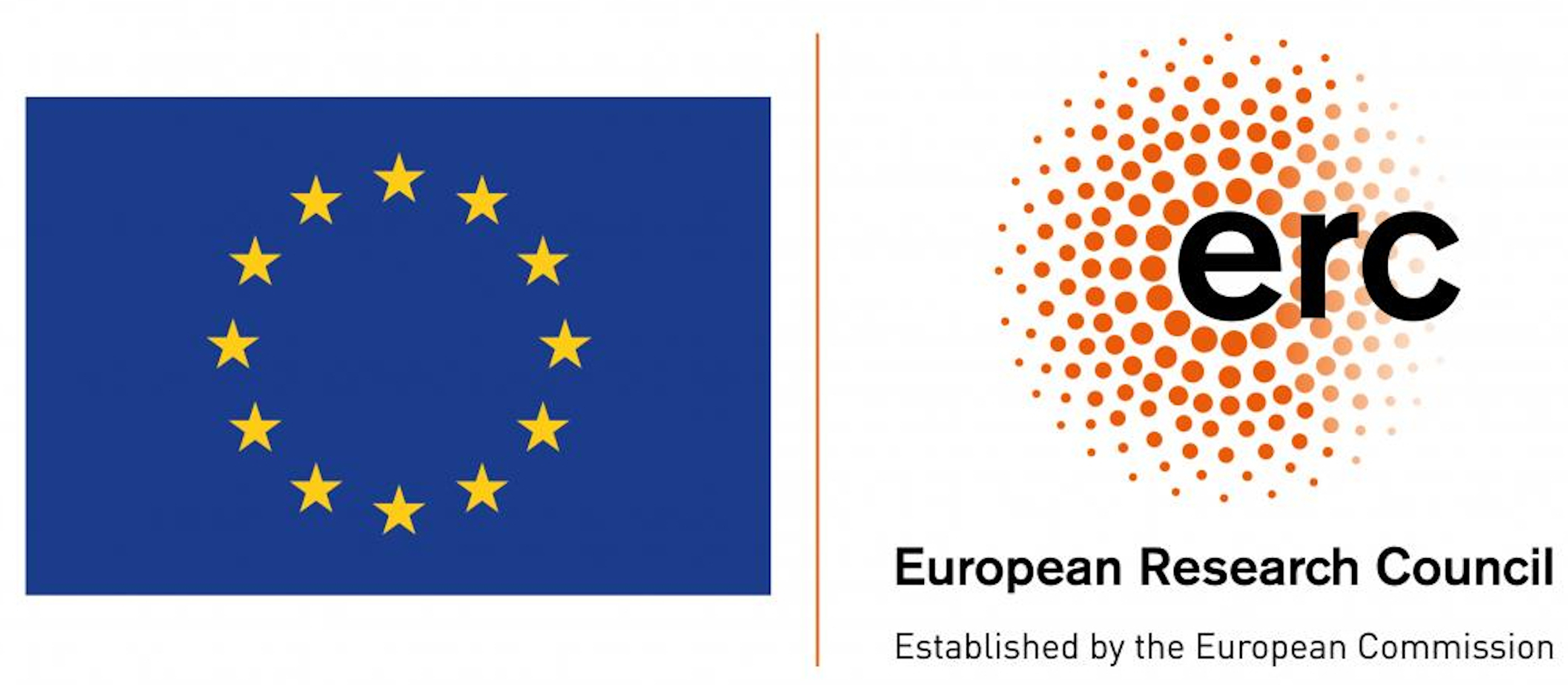}%
\end{textblock}

\newcommand{\SetCover}{\textsc{Set Cover}}
\newcommand{\GeometricSetCover}{\textsc{Geometric Set Cover}}
\newcommand{\SegmentSetCover}{\textsc{Segment Set Cover}}
\newcommand{\WeightedSegmentSetCover}{\textsc{Weighted Segment Set Cover}}
\newcommand{\WeightedGeometricSetCover}{\textsc{Weighted Geometric Set Cover}}
\newcommand{\true}{\texttt{true}}
\newcommand{\false}{\texttt{false}}

\thispagestyle{empty}

\begin{abstract}
	We study parameterized and approximation
	algorithms for a variant of $\SetCover$,
	where the universe of elements to be covered consists of points in the plane
	and the sets with which the points should be covered are segments.
	We call this problem $\SegmentSetCover$.
	We also consider a relaxation of the problem called {\em{$\delta$-extension}},
	where we need to cover the points by segments
	that are extended by a tiny fraction,
	but we compare the solution's quality to the optimum without~extension.

	For the unparameterized variant, we prove that $\SegmentSetCover$ does not admit a PTAS unless $\mathsf{P}=\mathsf{NP}$,
	even if we restrict segments to be axis-parallel
	and allow $\frac{1}{2}$-extension. On the other hand, we show that parameterization helps for the tractability of $\SegmentSetCover$: we give an FPT algorithm for unweighted $\SegmentSetCover$
	parameterized by the solution size $k$, a parameterized approximation scheme for $\WeightedSegmentSetCover$ with $k$ being the parameter, and an FPT algorithm for $\WeightedSegmentSetCover$ with $\delta$-extension parameterized by $k$ and $\delta$. In the last two results, relaxing the problem is probably necessary: we prove that
	\textsc{Weighted} \textsc{Segment} \textsc{Set} \textsc{Cover} without any relaxation
	is $\mathsf{W}[1]$-hard and, assuming ETH, there does not exist an algorithm running
	in time $f(k)\cdot n^{o(k / \log k)}$. This holds even if one restricts attention to axis-parallel segments.
\end{abstract}


\newpage

\section{Introduction}

\setcounter{page}{1}

In the classic \textsf{$\SetCover$} problem, we are given a set of elements (universe)
$\points$ and~a~family of sets $\sets$ that are subsets of $\points$ and sum up to the whole $\points$.
The~task is to find a subfamily $\sol \subseteq \sets$
such that $\bigcup \sol = \points$ and the size of $\sol$ is minimum possible.

In the most general form, $\SetCover$ is $\mathsf{NP}$-complete,
inapproximable within factor $(1-\delta)\ln |\points|$ for any $\delta>0$ assuming $\mathsf{P}\neq \mathsf{NP}$~\cite{set_cover_inapproximation},
and $\mathsf{W}[2]$-complete with the natural parameterization by the solution size~\cite[Theorem 13.21]{platypus_book}.
However, restricting the problem to various specialized settings
can lead to more tractable special cases. Particularly well-studied setting is that of $\GeometricSetCover$, where $\points$ consists of points in some Euclidean space $V$ (most often the plane $\mathbb{R}^2$), while $\sets$ consists of various geometric objects in $V$. In this paper we take a closer look at the $\SegmentSetCover$ problem, where we assume that $\points$ is a finite set of points in the plane and $\sets$ consists of segments in the plane (not necessarily axis-parallel). Each of these problems has also a natural weighted variant, where each set $A\in \sets$ comes with a nonnegative real weight $\mathbf{w}(A)$ and the task is to find a solution with the minimum possible total~weight.

\paragraph*{Approximation.}
Over the years there has been a lot of work related to approximation
algorithms for $\GeometricSetCover$. Notably,
$\GeometricSetCover$ with unweighted unit disks or weighted unit squares admits a PTAS~\cite{unit_disks,ErlebachL10}. When we consider the same problem
with weighted disks or squares (not necessarily unit), the problem admits a QPTAS
\cite{settling_apx_hardness}, see also \cite{voronoi_true}.
On the other hand, Chan and Grant
proved that unweighted $\GeometricSetCover$ with axis-parallel fat rectangles
is APX-hard~\cite{rectangles_apx_hard}. They also showed similar hardness
for $\GeometricSetCover$ with many other standard geometric objects. See the introductory section of~\cite{rectangles_apx_hard} for a wider discussion of approximation algorithms for $\GeometricSetCover$ with various kinds of geometric objects.

\paragraph*{Parameterization.}
We consider $\GeometricSetCover$
parameterized by the size of solution: Given an instance $(\points,\sets)$ and a parameter $k$, the task is to decide whether there is a solution of cardinality at most~$k$. In the weighted setting, we look for a minimum-weight solution among those of cardinality at most $k$, and $k$ remains a parameter.

(Unweighted) $\GeometricSetCover$ where $\sets$ consists of lines in the plane is called {\sc{Point Line Cover}}, and it is a textbook example of a problem that admits a quadratic kernel and a $2^{\Oh(k\log k)}\cdot n^{\Oh(1)}$-time fixed-parameter algorithm (cf.~\cite[Exercise~2.4]{platypus_book}). See also the work of Kratsch et al.~\cite{KratschPR16} for a matching lower bound on the kernel size and a discussion of the relevant literature. The simple branching and kernelization ideas behind the parameterized algorithms for {\sc{Point Line Cover}} were generalized by Langerman and Morin~\cite{LangermanM05} to an abstract variant of $\GeometricSetCover$ where the sets of $\sets$ can be assigned a suitable notion of dimension. This framework in particular applies to the problem of covering points with hyperspaces in $\mathbb{R}^d$.

As proved by Marx,
unweighted $\GeometricSetCover$ with unit squares in the plane is already $\mathsf{W}[1]$-hard~\cite[Theorem~5]{marx05}. Later, Marx and Pilipczuk showed that there is an~algorithm running in time $n^{\mathcal{O}(\sqrt{k})}$
that solves weighted $\GeometricSetCover$ with squares or with disks,
and that this running time is tight under the Exponential-Time Hypothesis (ETH)~\cite{voronoi}. However, they also showed that any small deviations from the setting of squares or disks --- for instance considering thin rectangles or rectangles with sidelengths in the interval $[1,1+\delta]$ for any $\delta>0$ --- lead to problems for which there are lower bounds refuting running times of the form $f(k)\cdot n^{o(k)}$ or $f(k)\cdot n^{o(k/\log k)}$, for any computable $f$. See~\cite{voronoi} for a broader exposition of these results and for more literature pointers.

We are not aware of any previous work that concretely considered the $\SegmentSetCover$ problem. In particular, it seems that the framework of Langerman and Morin~\cite{LangermanM05} does {\em{not}} apply to this problem, since no suitable notion of dimension can be assigned to segments in the plane (more concretely, the fundamental~\cite[Lemma~1]{LangermanM05} fails, which renders further arguments not applicable). In~\cite{Marx06} Marx considered the related {\sc{Dominating Set}} problem in intersection graphs of axis-parallel segments, and proved it to be $\mathsf{W}[1]$-hard. The parameterized complexity of the {\sc{Independent Set}} problem for segments in the plane was studied in the same work of Marx, and independently by K\'ara and Kratochv\'il~\cite{KaraK06}.

%

\paragraph*{$\delta$-extension.}
We also consider the $\delta$-extension relaxation of the $\SegmentSetCover$ problem. Formally, for a center-symmetric object $L\subseteq \mathbb{R}^2$ with
center of symmetry $S = (x_s, y_s)$,
the {\em{$\delta$-extension}} of $L$ is the~set:
$$L^{+\delta} = \{(1 + \epsilon)\cdot(x - x_s, y - y_s) + (x_s, y_s) : (x, y) \in L, 0 \le \epsilon < \delta\}.$$
That is, $L^{+\delta}$ is the image of $L$ under homothety centred
at $S$ with scale $(1+\delta)$ but with the extreme points excluded.
In particular, $\delta$-extension turns a closed segment
into a segment without endpoints
and a rectangle into the interior of a rectangle; this is a technical detail that will be useful in presentation.

In $\GeometricSetCover$ with $\delta$-extension, we assume that in the given instance~$(\points,\sets)$, $\sets$ consists of center-symmetric objects, and we are additionally given the accuracy parameter~$\delta>0$. The task is to find $\sol\subseteq \sets$ such that $\sol^{+\delta}\coloneqq \{L^{+\delta}\colon L\in \sol\}$ covers all points in~$\points$, but the quality of the solution --- be it the cardinality or the weight of $\sol$ --- is compared to the optimum without assuming extension. Thus, requirements on the the output solution are relaxed: the points of $\points$ have to be covered only after expanding every object of $\sol$ a tiny bit. The parameterized variants of $\GeometricSetCover$ with $\delta$-extension are defined naturally: the task is to either find a solution of size at most $k$ that covers all of $\points$ after $\delta$-extension, or conclude that there is no solution of size $k$ that covers $\points$ without extension.

The study of the $\delta$-extension relaxation is motivated by the {\em{$\delta$-shrinking}} relaxation considered in the context of the \textsc{Geometric Independent Set} problem: given a family $\sets$ of objects in the plane, find the maximum size subfamily of pairwise disjoint objects. In the $\delta$-shrinking model, the output solution is required to be disjoint only after shrinking every object by a $1-\delta$ multiplicative factor. \textsc{Geometric Independent Set} remains $\mathsf{W}[1]$-hard for as simple objects as unit disks or unit squares~\cite{Marx06} and admits a QPTAS for polygons~\cite{AdamaszekHW19}, but the existence of a PTAS for the problem is widely open. However, as first observed by Adamaszek et al.~\cite{shrinking_original}, and then confirmed by subsequent works of Wiese~\cite{shrinking2} and of~Pilipczuk et al.~\cite{shrinking1}, adopting the $\delta$-shrinking relaxation leads to a robust set of FPT algorithms and (efficient or parameterized) approximation schemes. The motivation of this work is to explore if the analogous $\delta$-extension relaxation of $\GeometricSetCover$ also leads to more positive results.

In fact, we are not the first to consider the $\delta$-extension relaxation of $\GeometricSetCover$. In~\cite{harpeled12}, Har-Peled and Lee considered the $\WeightedGeometricSetCover$ problem with $\delta$-extension\footnote{We note that Har-Peled and Lee considered a different definition of $\delta$-extension, where every object $L$ is extended by all points at distance at most $\delta\cdot \mathsf{rad}(L)$, where $\mathsf{rad}(L)$ is the radius of the largest circle inscribed in $L$. This definition works well for fat polygons, but not so for segments, hence we adopt the homothetical definition of $\delta$-extension discussed above.} for fat polygons, and proved that the problem admits a PTAS with running time $|\sets|^{\mathcal{O}(\epsilon^{-2}\delta^{-2})}\cdot |\points|$. Given this result, our goal is to understand the complexity in the setting of ultimately non-fat polygons: segments.

\paragraph*{Our contribution.}
First, we show that $\SegmentSetCover$ does not have a polynomial-time approximation scheme (PTAS) assuming $\mathsf{P}\neq \mathsf{NP}$,
even if segments are axis-parallel and we relax the problem with $\frac{1}{2}$-extension. Thus, there is no hope for the analog of the result of Har-Peled and Lee~\cite{harpeled12} in the setting of~segments.

\begin{restatable}{theorem}{segmentCoverApxHard}{
\label{segment_cover_apx_hard}
	There exists a constant $\gamma>0$ such that, unless $\mathsf{P}= \mathsf{NP}$, there is no polynomial-time algorithm that given an instance $(\points,\sets)$ of (unweighted) $\SegmentSetCover$ in which all segments are axis-parallel, returns a set $\sol\subseteq \sets$ such that $\sol^{+\frac{1}{2}}$ covers $\points$ and $|\sol|\leq (1+\gamma)\cdot \opt$, where $\opt$ denotes the minimum size of a subset of $\sets$ that covers $\points$.
}\end{restatable}
%
%
%

Theorem~\ref{segment_cover_apx_hard} justifies also considering parameterization by the solution size $k$.
For this parameterization, we provide three parameterized algorithms:
\begin{itemize}
 \item an FPT algorithm for (unweighted) $\SegmentSetCover$ with $k$ being the parameter;
 \item a {\em{parameterized approximation scheme}} (PAS) for $\WeightedSegmentSetCover$: a $(1+\epsilon)$-approximation algorithm with running time of the form $f(k,\epsilon)\cdot (|\points||\sets|)^{\Oh(1)}$; and
 \item an FPT algorithm for $\WeightedSegmentSetCover$ with $\delta$-extension, where both $k$ and $\delta>0$ are the parameters.
\end{itemize}
Formal statements of these results follow below.

\begin{restatable}{theorem}{segmentCoverFpt}{
	\label{segment_cover_fpt}
	There is an algorithm that given a family $\sets$ of
	segments in the plane,
	a set $\points$ of points in the plane,
	and a parameter $k$,
	runs in time ${k^{\mathcal{O}(k)}\cdot (|\points||\sets|)^{\mathcal{O}(1)}}$,
	and either outputs a set $\sol \subseteq \sets$
	such that $|\sol| \le k$ and $\sol$ covers all points in~$\points$,
	or determines that such a set $\sol$ does not exist.
}\end{restatable}

\begin{restatable}{theorem}{pasWeightedSegment}{
	\label{pas_weighted_segment}
	There is an algorithm that given a family $\sets$ of
	weighted segments in the plane,
	a set $\points$ of points in the plane, and parameters $k$ and $\epsilon > 0$,
	runs in time $(k/\epsilon)^{\Oh(k)}\cdot (|\points||\sets|)^{\Oh(1)}$ and
	outputs a set $\sol$ such that:\begin{itemize}
	\item $\sol \subseteq \sets$, $|\sol| \le k$, and $\sol$ covers all points in $\points$, and
	\item the weight of $\sol$ is not greater than $1+\epsilon$ times the minimum weight
	of a subset of $\sets$ of size at most $k$
	that covers $\points$,
	\end{itemize}
	or determines that there is no set $\sol\subseteq \sets$ with $|\sol| \le k$
	such that $\sol$ covers all points in $\points$.
}\end{restatable}

\begin{restatable}{theorem}{fptWeightedSegment}{
	\label{fpt_weighted_segment}
	There is an algorithm that given a family $\sets$ of
	weighted segments in the plane,
	a set $\points$ of points in the plane, and parameters $k$ and $\delta > 0$,
	runs in time $f(k, \delta) \cdot (|\points||\sets|)^{\mathcal{O}(1)}$ for some computable function $f$ and
	outputs a set $\sol$ such that:\begin{itemize}
	\item $\sol \subseteq \sets$, $|\sol| \le k$, $\sol^{+\delta}$ covers all points in $\points$, and
	\item the weight of $\sol$ is not greater than the minimum weight
	of a subset of $\sets$ that covers $\points$ without $\delta$-extension,
	\end{itemize}
	or determines that there is no set $\sol\subseteq \sets$ with $|\sol| \le k$
	such that $\sol$ covers all points in $\points$.
}\end{restatable}

It is natural to ask whether relying on relaxations --- $(1+\epsilon)$-approximation or $\delta$-extension --- is really necessary for $\WeightedSegmentSetCover$, as Theorem~\ref{segment_cover_fpt} shows that it is not in the unweighted setting. Somewhat surprisingly, we show that this is the case by proving the following result. Recall that here we consider $\WeightedSegmentSetCover$ as a parameterized problem where we seek a solution of the minimum total weight among those of cardinality at most $k$.

%

\begin{restatable}{theorem}{wOneHard}
\label{w1_hard}
	The $\WeightedSegmentSetCover$ problem is $\mathsf{W}[1]$-hard when parameterized by $k$ and
	assuming ETH, there is no algorithm for this
	problem with running time
	$f(k)\cdot(|\points| + |\sets|)^{o(k/\log k)}$
	for any computable function $f$.
	Moreover, this holds even if all segments in $\sets$
	are axis-parallel.
\end{restatable}

Thus, the uncovered parameterized complexity of $\SegmentSetCover$ is an interesting one: the problem is FPT when parameterized by the solution size $k$ in the unweighted setting, but this tractability ceases to hold when moving to the weighted setting. However, fixed-parameter tractability in the weighted setting can be restored if one considers any of the following relaxations: $(1+\epsilon)$-approximation or $\delta$-extension.

\paragraph*{Organization.} In Section~\ref{sec:fpt} we prove Theorems~\ref{segment_cover_fpt},~\ref{pas_weighted_segment} and~\ref{fpt_weighted_segment}, in Section~\ref{chapter:w1_hard} we prove Theorem~\ref{w1_hard}, and in Section~\ref{section:apx_hard} we prove Theorem~\ref{segment_cover_apx_hard}.

\section{Algorithms}\label{sec:fpt}

In this section we give our positive results --- Theorems~\ref{segment_cover_fpt},~\ref{pas_weighted_segment}, and~\ref{fpt_weighted_segment}.
We start with a shared definition.
	For a set of collinear points $C$ in the plane,
	{\em{extreme points}} of $C$ are the endpoints
	of the smallest segment that covers all points from set $C$. In particular, if $C$ consists of one point or is empty, then
	there are $1$ or $0$ extreme points, respectively.

\subsection{Unweighted segments and a parameterized approximation scheme}
\label{segments_in_arbitrary_direction}
We first a give an FPT algorithm for $\WeightedSegmentSetCover$ where we additionally consider the number of different weights to be the parameter.

\begin{theorem}\label{thm:few_weights}
 	There is an algorithm that given a family $\sets$ of
	weighted segments in the plane,
	a set $\points$ of points in the plane,
	and a parameter $k$,
	runs in time ${(qk)^{\mathcal{O}(k)}\cdot (|\points||\sets|)^{\mathcal{O}(1)}}$, where $q$ is the number of different weights used by the weight function,
	and either outputs a solution $\sol \subseteq \sets$
	such that $|\sol| \le k$ and $\sol$ covers all points in~$\points$,
	or determines that such a set $\sol$ does not exist.
\end{theorem}

Clearly, Theorem~\ref{segment_cover_fpt} follows from applying Theorem~\ref{thm:few_weights} for $q=1$. However, later we use Theorem~\ref{thm:few_weights} for larger values of $q$ to obtain our parameterized approximation scheme: Theorem~\ref{pas_weighted_segment}.

We remark that the proof of Theorem~\ref{thm:few_weights} relies on branching and kernelization arguments that are standard in parameterized algorithms. Even though the statement does not formally follow from the work of Langerman and Morin~\cite{LangermanM05}, the basic technique is very similar.

Towards the proof of Theorem~\ref{thm:few_weights}, we may assume that the instance $(\points,\sets,\mathbf{w})$ given on input, where $\mathbf{w}\colon \sets\to \mathbb{R}_{\geq 0}$  denotes the weight function on $\sets$, is {\em{reasonable}} in the following sense: there do not exist distinct $A,B\in \sets$ with the same weight such that $A\cap \points \subseteq B\cap \points$. Indeed, then $A$ could be safely removed from $\sets$, since in any solution, taking $B$ instead of $A$ does not increase the weight and may only result in covering more points in $\points$.
In the next lemma we show that in reasonable instances we can find a small subset of $\sets$ that is guaranteed to intersect every small solution.

\newcommand{\Rr}{\mathcal{R}}

\begin{lemma}
	\label{fpt_long_lines}
	Suppose $(\points,\sets,\mathbf{w})$ is a reasonable instance
	of $\WeightedSegmentSetCover$ where the weight function $\mathbf{w}$ uses at most $q$ different values. Suppose further that there exists a line $L$ in the plane with at least
	$k+1$ points of $\points$ on it. Then there exists a subset $\Rr \subseteq \sets$
	of size at most $qk$
	such that every subset $\sol\subseteq \sets$ with $|\sol| \le k$ that covers $\points$
	satisfies $|\Rr \cap \sol| \ge 1$.
	Moreover, such a subset $\Rr$ can be found in~polynomial time.
\end{lemma}
\begin{proof}
Let us enumerate the points of $\points$ that lie on $L$ as $x_1, x_2, \ldots, x_t$
in the order in which they appear on $L$. By reasonability of $(\points,\sets)$, for every $i\in \{1,\ldots,k\}$ there exist at most $q$ different segments in $\sets$ that are collinear with $L$ and cover $x_i$, but do not cover $x_{i-1}$ (or just cover $x_1$, in case $i=1$). Indeed, if $A\in \sets$ is collinear with $L$, covers $x_i$ and does not cover $x_{i-1}$, then $A\cap \points = \{x_i,\ldots,x_j\}$ for some $j\geq i$; so if there was another $B\in \sets$ with the same property and the same weight as $A$, then the reasonability of $(\points,\sets)$ would imply that $A=B$.
Let $\Rr_i$ be the set of segments with the property discussed above; then $|\Rr_i|\leq q$.
Our proposed set is defined as:
$$\Rr\coloneqq  \bigcup_{i=1}^k \Rr_i.$$
Clearly, $\Rr$ can be found in polynomial time and $|\Rr|\leq qk$.
It remains to prove that $\Rr$ has the desired property. Consider any set $\sol\subseteq \sets$ of size at most $k$ that covers $\points$.

Let $\sol_L$ be the set of segments from $\sol$ that are collinear with $L$.
Every segment that is not collinear with $L$ can cover at~most one of
the points that lie on~this line.
Hence, if $\sol_L$ was empty, then
$\sol$ would cover at most $k$ points on line $L$,
but $L$ had at least $k+1$ different points from $\points$ on it.

Therefore, we know that $\sol_L$ is not empty
and hence $|\sol - \sol_L| \le k-1$.
Segments from $\sol - \sol_L$ can cover at most $k-1$
points among $\{x_1, x_2, \ldots, x_k\}$, therefore at least
one of these points must be covered by segments from $\sol_L$.
Let $i\in \{1,\ldots,k\}$ be the smallest index such that $x_i$ is covered by a segment in $\sol_L$. Then, by minimality, this segment cannot cover $x_{i-1}$ (if existent), so it must belong to $\Rr_i$. We conclude that $\Rr \cap \sol$ is nonempty, as desired.
\end{proof}

With
Lemma~\ref{fpt_long_lines} in hand, we prove Theorem~\ref{thm:few_weights} using a straightforward branching strategy.

\begin{proof}[Proof of Theorem~\ref{thm:few_weights}.]
Let $(\points,\sets,\mathbf{w})$ be the given instance and $k$ be the given parameter: the target size of a solution. We present a recursive algorithm that proceeds as follows:
\begin{enumerate}[label={(\arabic*)}]
\item As long as there are distinct sets $A,B\in \sets$ with $A\cap \points\subseteq B\cap \points$ and $\mathbf{w}(A)=\mathbf{w}(B)$, remove $A$ from $\sets$. Once this step is applied exhaustively, the instance $(\points,\sets,\mathbf{w})$ is~reasonable.
\item \label{step2} If there is a line with at least $k+1$ points from $\points$,
we branch over the choice of adding to~the~solution
one of~the~at~most $qk$ possible segments from the set $\Rr$
provided by Lemma~\ref{fpt_long_lines}. That is, for every $s\in \Rr$, we recurse on the instance $(\points - s, \sets - \{s\},\mathbf{w})$,
and parameter $k-1$. If any such recursive call returned a solution $\sol'$, then return the lightest among solutions $\sol'\cup \{s\}$ obtained in this way.  Otherwise, return that there is no~solution.
\item \label{step3} If every line has at most $k$ points on it and $|\points| > k^2$,
then return that there is no solution.
\item If $|\points| \le k^2$, solve the problem by brute force:
check all subsets of $\sets$ of size at most $k$.
\end{enumerate}

That the algorithm is correct is clear: the correctness of step \ref{step2} follows from Lemma~\ref{fpt_long_lines}, and to see the correctness of step \ref{step3} note that if no line contains more than $k$ points, than no segment of $\sets$ can cover more than $k$ points in $\points$, hence having more than $k^2$ points in $\points$ implies that there is no solution of size at most $k$.

For the time complexity, observe that in the leaves of the recursion we have $|\points| \le k^2$, so $|\sets| \le qk^4$,
because every segment can be uniquely identified by its weight and the two extreme points it covers
(this follows from reasonability). Therefore, there are $\binom{qk^4}{\leq k}\leq (qk)^{\Oh(k)}$
possible solutions to check, each can be checked in polynomial time.
Thus, step (4) takes time $(qk)^{\Oh(k)}$ whenever applied in the leaf of the~recursion.

During the recursion, the parameter $k$ is decreased with every
recursive call, so the recursion tree has depth at most $k$ and at each node we branch over at most $qk$ possibilities. Thus, there are at most $\Oh((qk)^k)$ nodes in the recursion tree, and local computation in each of them can be done in time $(|\points||\sets|)^{\Oh(1)}\cdot (qk)^{\Oh(k)}$ (the second factor is due to possibly applying step (4) in the leaves). Thus, the time complexity of the algorithm is $(qk)^{\Oh(k)}\cdot (|\points||\sets|)^{\Oh(1)}$.
\end{proof}

Finally, we use Theorem~\ref{thm:few_weights} to prove Theorem~\ref{pas_weighted_segment}, recalled below for convenience. The idea is to multiplicatively round the weights so that we obtain an instance with only few different weight values, on which the algorithm of Theorem~\ref{thm:few_weights} can be employed.

\pasWeightedSegment*

\begin{proof}
Let $\sol^\star$ be an optimum solution: a minimum-weight set at most $k$ segments in $\sets$ that covers $\points$. The algorithm does not know $\sol^\star$, but by branching into at most $|\sets|$ choices we may assume that it knows the weight of the heaviest segment in $\sol^\star$; call this weight $W$. Thus, we have $W\leq \mathbf{w}(\sol^\star)\leq kW$. We may dispose of all segments in $\sets$ whose weight is larger than $W$, as they will for sure not participate in the~solution.

 We define a new weight function $\mathbf{w}'\colon \sets\to \mathbb{R}_{\geq 0}$ as follows. Consider any segment $A\in \sets$. Provided $\mathbf{w}(A)\leq \frac{\epsilon}{2k}\cdot W$, set $\mathbf{w}'(A)\coloneqq \frac{\epsilon}{2k}\cdot W$. Otherwise, set $\mathbf{w}'(A)\coloneqq \frac{W}{(1+\epsilon/2)^i}$, where $i$ is the unique integer such that
 $$\frac{W}{(1+\epsilon/2)^{i+1}}<\mathbf{w}(A)\leq \frac{W}{(1+\epsilon/2)^i}.$$
 Note that the assumption $\mathbf{w}(A)>\frac{\epsilon}{2k}\cdot W$ implies that we always have $i\leq \log_{1+\epsilon/2} (2k/\epsilon)=\Oh(1/\epsilon \log (k/\epsilon))$. As we also have $i\geq 0$ due to removing segments of weight larger than~$W$, we conclude that the weight function $\mathbf{w}'$ uses at most $\Oh(1/\epsilon \log (k/\epsilon))$ different weight values.

 Next, observe that for every segment $A\in \sets$, we have
 $$\mathbf{w}'(A)\leq (1+\epsilon/2)\cdot \mathbf{w}(A)+\frac{\epsilon}{2k}\cdot W.$$
 Summing this inequality through all segments of $\sol^\star$ yields
 $$\mathbf{w}'(\sol^\star) \leq (1+\epsilon/2)\cdot \mathbf{w}(\sol^\star) + k\cdot \frac{\epsilon}{2k}\cdot W\leq (1+\epsilon/2)\cdot \mathbf{w}(\sol^\star) + \epsilon/2\cdot \mathbf{w}(\sol^\star) = (1+\epsilon)\cdot \mathbf{w}(\sol^\star).$$
 As $\sol^\star$ is an optimum solution, we conclude that the optimum solution in the instance $(\points,\sets,\mathbf{w}')$ for parameter $k$ is at most $(1+\epsilon)$ times heavier than the optimum solution in the instance $(\points,\sets,\mathbf{w})$ for parameter $k$. Hence, it suffices to apply the algorithm of Theorem~\ref{thm:few_weights} to the instance $(\points,\sets,\mathbf{w}')$ and parameter $k$ and return the obtained solution. The running time is $(1/\epsilon \cdot k\log (k/\epsilon))^{\Oh(k)}\cdot (|\points||\sets|)^{\Oh(1)}=(k/\epsilon)^{\Oh(k)}\cdot (|\points||\sets|)^{\Oh(1)}$, as promised.
\end{proof}

\subsection{Weighted segments with $\delta$-extension}
\label{section:fpt_weighted}

In this section we prove Theorem~\ref{fpt_weighted_segment}, restated below for convenience.

\fptWeightedSegment*


Roughly speaking, our approach to prove Theorem~\ref{fpt_weighted_segment} is to find a small kernel for the problem; but we need to be careful with the definition of kernelization, because we work in the $\delta$-extension model. The key technical tool will be the notion of a {\em{dense subset}}.

\paragraph*{Dense subsets.}
Intuitively speaking, for a set of collinear points $C$, a subset $A\subseteq C$ is dense if any small cover of $A$ becomes a cover of $C$ after a tiny extension. This is formalized in the following definition.

\begin{definition}
	For a set of collinear points $C$,
	a subset $A \subseteq C$ is {\em{$(k,\delta)$-dense}} in $C$
	if for any set of segments $\Rr$ that covers $A$ and
	such that $|\Rr| \le k$, it holds that $\Rr^{+\delta}$ covers $C$.
\end{definition}

The key combinatorial observation in our approach is expressed in the following Lemma~\ref{dense_set_exists}: in every collinear set $C$ one can always find a $(k,\delta)$-dense subset of size bounded by a function of $k$ and $\delta$.
Later, this lemma will allow us to find a kernel
for our original problem.

\begin{lemma}
	\label{dense_set_exists}
	For every set $C$ of collinear points in the plane, $\delta > 0$ and $k \ge 1$,
	there exists a $(k,\delta)$-dense set $A \subseteq C$ of size
	at most $(2+\frac{4}{\delta})^k$.
	Moreover, such a $(k,\delta)$-dense set
	can be computed in time $\Oh(|C| \cdot (2+\frac{4}{\delta})^k)$.
\end{lemma}

\begin{proof}
We give a proof of the existence of such a dense subset $A$, and at the end we will argue that the proof naturally gives rise to an algorithm with the promised complexity. We fix $\delta$ and proceed by induction on $k$. Formally, we shall prove the following stronger statement:
For any set of collinear points $C$, there exists a subset $A\subseteq C$ such that:
\begin{itemize}
\item $A$ is $(k, \delta)$-dense in $C$,
\item $|A| \le (2+\frac{4}{\delta})^k$, and
\item the extreme points of $C$ are in $A$.
\end{itemize}

Consider first the base case when $k=1$.
Then it is sufficient to just take $A$ that consists of the (at most~$2$) extreme points of $C$.
Indeed, if the extreme points of $C$ are covered with one segment, then this segment must cover the whole set $C$ (even without extension).
Thus, the set $A$ has size at most $2 < (2+\frac{4}{\delta})^1$, as required.

We now proceed to the inductive step.
Assuming inductive hypothesis for any set of collinear points $C$
and for parameter $k$, we will prove it for $k+1$.

Let $s$ be the minimal segment that includes all points from $C$.
That is, $s$ is the segment whose endpoints are the extreme points of $C$.
Split $s$ into $M \coloneqq \lceil1+\frac{4}{\delta}\rceil$ subsegments of equal length.
We name these segments $v_1,v_2,\ldots,v_M$ in order, and we consider them closed. Note that
$|v_i| = \frac{|s|}{M}$ for each $1 \le i \le M$, where $|\cdot|$ denotes the length of a segment.
Let $C_i$ be the subset of $C$ consisting of points belonging to~$v_i$.
Further, let $t_i$ be the segment with endpoints being the extreme points of $C_i$.
Note that $t_i$ might be a degenerate single-point segment if $C_i$ consists of one point,
or even $t_i$ might be empty if $C_i$ is empty.
Figure~\ref{fig:fpt_v_f_def} presents an example of the~construction.

\begin{figure}[h]
\begin{center}
\def\svgwidth{\textwidth}
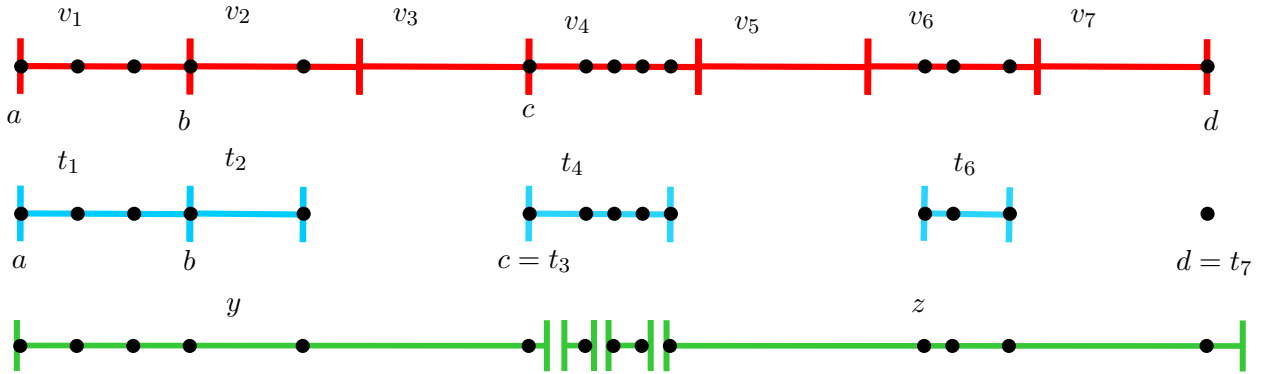
\caption{Example of the construction in the proof of Lemma~\ref{dense_set_exists} for $M = 7$ and some set of points $C$ (marked with black circles).
The top panel shows segments $v_i$. The middle panel shows segments~$t_i$. Note that $t_5$ is an empty segment,
because there are no points in $C$ that belong to $v_5$, while
each of the segments $t_3$ and $t_7$ is degenerated to a single point:
$c$ and $d$, respectively.
Segments $t_1$ and $t_2$ share one point $b$.
The bottom panel shows an example of the second case in the correctness proof: a solution $\Rr$ of size $4$ whose all segments intersect $t_4$. Then one of $y$ and $z$ will cover the whole of $C_4$ after extension.
}
\label{fig:fpt_v_f_def}
\end{center}
\end{figure}

We use the inductive hypothesis to choose a $(k, \delta)$-dense subset $A_i$ of $C_i$, for each $i\in \{1,\ldots,M\}$. Note that if $|C_i| \le 1$, then $A_i = C_i$, so $A_i$ is $(k, \delta)$-dense set for $C_i$. Also, by assumption, $A_i$ contains the extreme points of $C_i$.

Next, we define $A \coloneqq \bigcup_{i=1}^{M} A_i$.
Thus $A$ includes the extreme points of $C$,
because they are included in the sets $A_1$ and $A_M$.
By induction, the size of each $A_i$ is at most $(2+\frac{4}{\delta})^{k}$. Therefore,
$$|A|\leq M\left(2+\frac{4}{\delta}\right)^{k} =
\left\lceil1+\frac{4}{\delta}\right\rceil\cdot\left(2+\frac{4}{\delta}\right)^{k}
\le \left(2+\frac{4}{\delta}\right)^{k+1}.$$
We are left with verifying that $A$ is $(k+1,\delta)$-dense in $C$.
For this, consider any cover of $A$ with $k+1$ segments and call it $\Rr$.

Consider any segment $t_i$. If there exists a segment $x\in \Rr$
that is disjoint with $t_i$,
then $\Rr- \{x\}$ constitutes a cover of $A_i$ with at most $k$
segments.
Since $A_i$ is $(k, \delta)$-dense in $C_i$,
$(\Rr - \{x\})^{+\delta}$ covers $C_i$.
So $\Rr^{+\delta}$ covers $C_i$ as well.

On the other hand, if for any fixed $t_i$ a segment $x\in \Rr$ as above
does not exist, then all the $k+1$ segments of $\Rr$ intersect $t_i$.
An example of such a situation is depicted in the bottom panel of Figure~\ref{fig:fpt_v_f_def}.
Let us consider any such $t_i$.
By induction, the endpoints of $s$ are
in $A_1$ and $A_M$ respectively, so $\Rr$ must cover them.
So for each endpoint of $s$, there exists
a segment in $\Rr$ that contains this endpoint and intersects $t_i$.
Let us call these two segments $y$ and~$z$. It follows that
$|y| + |z| + |t_i| \ge |s|$.
Since $|t_i| \le |v_i| = \frac{|s|}{M} \le \frac{|s|}{1+\frac{4}{\delta}} = \frac{\delta|s|}{\delta+4}$,
we have $$\max(|y|, |z|) \ge |s|\left(1-\frac{\delta}{\delta+4}\right)/2 = \frac{2|s|}{\delta+4}.$$
%
After $\delta$-extension, the longer of the segments $y$ and $z$ will
expand at both ends by at least:
$$\delta/2\cdot \max(|y|, |z|) \ge \frac{\delta|s|}{\delta+4} =
\frac{|s|}{1+\frac{4}{\delta}} \ge \frac{|s|}{M} = |v_i| \ge |t_i|.$$
Therefore, the longer of segments $y$ and $z$ will cover the whole segment $t_i$
after $\delta$-extension. We conclude that $\Rr^{+\delta}$ covers $C_i$ as well.

Since $C = \bigcup_{i=1}^M C_i$, we conclude that $\Rr^{+\delta}$ covers $C$. So indeed, $A$ is $(k+1,\delta)$-dense in~$C$. This concludes the proof of the existence of such a dense set $A$. To compute $A$ in time $\Oh\left(|C|\cdot\left(2+\frac{4}{\delta}\right)^k\right)$ observe that the inductive proof explained above can be easily turned into a recursive procedure that for a given $C$ and $k$, outputs a $(k,\delta)$-dense subset of $C$. The recursion tree of this procedure has size $\Oh\left(\left(2+\frac{4}{\delta}\right)^k\right)$ in total, while every recursive calls uses $\Oh(|C|)$ time for internal computation, so the total running time is $\Oh\left(|C|\cdot\left(2+\frac{4}{\delta}\right)^k\right)$.
\end{proof}

\paragraph*{Long lines.} We need a few additional observations in the spirit of the algorithm of Theorem~\ref{thm:few_weights}. For a finite set of points $\points$ in the plane, call a line $L$ {\em{$k$-long}} with respect to $\points$ if $L$ contains more than $k$ points from $\points$. We have the following observations.

\begin{lemma}
\label{few_long_lines}
Let $\points$ be a finite set of points in the plane such that there are more than $k$ lines that are $k$-long with respect to $\points$. Then
$\points$ cannot be covered with $k$ segments.
\end{lemma}
\begin{proof}
We proceed by contradiction.
Assume there are at least $k+1$ different
$k$-long lines and there is a set of segments $\Rr$ of size at most $k$
covering all points in $\points$.

Consider any $k$-long line $L$.
Note that every segment $\Rr$ which is not collinear with~$L$,
covers at most one point that lies on $L$.
Since $L$ is long, there are at least $k+1$ points from $\points$ that lie on $L$.
This implies that there must be a segment in $\Rr$ that is
collinear with~$L$.

Since we have at least $k+1$ different long lines,
there are at least $k+1$
segments in $\Rr$ collinear with different lines.
This contradicts the assumption that $|\Rr| \le k$.
\end{proof}

\begin{lemma}
\label{few_points}
Let $\points$ be a finite set of points in the plane such that there are more than $k^2$ points from $\points$
that do not lie on any line that is $k$-long with respect to $\points$. Then $\points$ cannot be covered with $k$ segments.
\end{lemma}
\begin{proof}
Assume  we have more than $k^2$ points
in $\points$ that do not lie on any $k$-long line. Call this set $A$.
Suppose, for contradiction, that there is a set of segments $\Rr$ of size at most $k$
that covers all points in~$\points$.

Since any line in the plane can cover only at most $k$ points in $A$, the same is also true for every segment in $\Rr$. Therefore, the segments from $\Rr$ can cover at most $k^2$ points in $A$ in total. As $|A|>k^2$, $\Rr$ cannot cover the whole $A$, which is a subset of $\points$; a contradiction.
\end{proof}

We are now ready to give a proof of Theorem \ref{fpt_weighted_segment}.

\newcommand{\instance}{(\points,\sets,\mathbf{w})}
\newcommand{\instancePrim}{(\points',\sets',\mathbf{w})}

\begin{proof}[Proof of Theorem \ref{fpt_weighted_segment}]
Let $\instance$ be the instance of $\WeightedSegmentSetCover$ given on input, where $\mathbf{w}\colon \sets\to \mathbb{R}_{\geq 0}$ is the weight function. Further, let $k$ and $\delta>0$ be the input parameters.
Our goal is to either conclude that $\instance$ has no solution of cardinality at most $k$, or to find an instance
$\instancePrim$ of size bounded by $f(k,\delta)$ for some computable function $f$ and satisfying $\points'\subseteq \points$ and $\sets'\subseteq \sets$, such that the following two properties hold:
\begin{itemize}
\item \textit{(Property 1)} For every set $\sol\subseteq \sets$ such that $|\sol|\leq k$ and $\sol$ covers $\points$,
there is a set $\sol_1 \subseteq \sets'$ such that
$|\sol_1| \le k$, the weight of $\sol_1$ is not greater than the weight of $\sol$,
and $\sol_1$ covers $\points'$.
\item \textit{(Property 2)}
For every set $\sol \subseteq \sets'$ such that $|\sol| \le k$
and $\sol$ covers all points in $\points'$, $\sol^{+\delta}$
covers all points in the original set $\points$.
\end{itemize}
Suppose we constructed such an instance $\instancePrim$.
Then using Property 1 we know that an optimum solution
of size at most $k$ to $\instancePrim$
has no greater weight than an optimum solution
of size at most $k$ to $\instance$.
On the other hand, using \textit{Property 2} we know that
any solution to $\instancePrim$
after $\delta$-extension covers $\points$. So it will remain to find an optimum solution to the instance $\instancePrim$. This can be done by brute-force in time $|\sets'|^{k+\Oh(1)}\cdot |\points'|^{\Oh(1)}$, which is bounded by a computable function of $k$ and $\delta$.


It remains to construct the instance $\instancePrim$.
Let $\ell$ be the number of different lines that are $k$-long with respect to $\points$.
By Lemmas~\ref{few_long_lines} and~\ref{few_points},
if we had more than $k$ different $k$-long lines
or more than $k^2$ points from $\points$
that do not lie on any $k$-long line, then we can safely conclude that $(\points,\sets,\mathbf{w})$ has no solution of cardinality at most $k$, and terminate the algorithm. So assume otherwise, in particular $\ell\leq k$.

Next, we cover $\points$ with at most $k+1$ sets:
\begin{itemize}
\item $D$ consists of all points in $\points$ that do not lie on any $k$-long line. Then we have $|D| \le k^2$.
\item For $1 \le i \le \ell$, $C_i$ consists of all points in $\points$ that lie on the $i$-th long line. Then $|C_i| > k$.
\end{itemize}
Note that sets $C_i$ do not need to be disjoint.

For every set $C_i$, we apply Lemma~\ref{dense_set_exists}
to obtain a subset $A_i\subseteq C_i$ that is $(k,\delta)$-dense and satisfies $|A_i| \le (2+\frac{4}{\delta})^k$.
We define $\points'\coloneqq D \cup \bigcup_{i=1}^\ell A_i$. Thus, $\points'$ has size at most
$k^2 + k(2+\frac{4}{\delta})^k$.
Further, we define $\sets'$ as follows: for every pair of points in $\points'$, if there are segments in $\sets$ that cover this pair of points,
we choose one such segment with the lowest weight and include it in $\sets'$.
Thus $\sets'$ has size at most~$|\points'|^2$,
which means that both $\sets'$ and $\points'$ have sizes bounded
by~$\Oh\left((k^2 + k(2+\frac{4}{\delta})^k)^2\right)$. We are left with verifying Properties 1 and~2.

For Property~2, consider any set $\sol \subseteq \sets'$ such that $|\sol| \le k$
and $\sol$ covers all points in $\points'$. Then in particular, for every $i\in \{1,\ldots,\ell\}$, $\sol$ in covers all points in $A_i$. As $A_i$ is $(k,\delta)$-dense in $C_i$, we conclude that $\sol^{+\delta}$ covers $C_i$. Hence $\sol^{+\delta}$ covers $D\cup \bigcup_{i=1}^\ell C_i=\points$, as required.

For Property~1, consider any solution
$\sol$ to $\instance$ of size at most $k$.
For every segment $s\in \sol$, let $B_s$ be the set of points in $\points'$ that are covered by $s$.
$B_s$ is of course a set of collinear points, hence
$B_s$ can be covered by any segment that covers the extreme points of~$B_s$.
Therefore, we can replace $s$ with a segment $s'\in \sets$
that has the lowest weight among the segments that cover the extreme points of~$B_s$.
Such a segment belongs to $\sets'$ by construction, and $s'$ has weight no greater than the weight of $s$, because $s$ also covers the extreme points of $B_s$. Therefore, if $\sol_1\subseteq \sets'$ is the set obtained by performing such replacement for every $s\in \sol$, then $\sol_1$ has both size and weight not greater than $\sol$, and $\sol_1$ covers $\points'$.
\end{proof}

\newcommand{\PartSubIso}{{\sc{Partitioned Subgraph Isomorphism}}}

\section{$\mathsf{W}[1]$-hardness of $\WeightedSegmentSetCover$}
\label{chapter:w1_hard}

In this section we prove Theorem~\ref{w1_hard}, recalled below for convenience.

\wOneHard*

To prove Theorem~\ref{w1_hard},
we give a reduction from a $\mathsf{W}[1]$-hard problem:
\PartSubIso, defined as follows. An instance of \PartSubIso{} consists of a {\em{pattern graph}} $H$, a {\em{host graph}} $G$, and a function $\lambda\colon V(G)\to V(H)$ that colors every vertex of $G$ with a vertex of $H$. The task is to decide whether there exists a subgraph embedding $\phi\colon V(H)\to V(G)$ that respects the coloring~$\lambda$. That is, the following conditions have to be satisfied.
\begin{itemize}
 \item $\lambda(\phi(a))=a$ for each $a\in V(H)$, and
 \item $\phi(a)$ and $\phi(b)$ are adjacent in $G$ for every edge $ab\in E(H)$.
\end{itemize}
The following complexity lower bound for \PartSubIso{} was proved by Marx in~\cite{Marx10beat}.

\begin{theorem}[\cite{Marx10beat}]
\label{thm:subisoHardness}
Consider the \PartSubIso{} problem where the pattern graph $H$ is assumed to be $3$-regular. Then this problem is $\mathsf{W}[1]$-hard when parameterized by $k$, the number of vertices of $H$, and assuming the ETH there is no algorithm solving this problem in time $f(k)\cdot |V(G)|^{o(k/\log k)}$, where $f$ is any computable function.
\end{theorem}

\newcommand{\instanceSetCover}{(\points, \sets, \mathbf{w})}

In the remainder of this section we prove Theorem~\ref{w1_hard}
by providing a parameterized reduction from \PartSubIso{} to $\WeightedSegmentSetCover$. The technical statement of the reduction is encapsulated in the following lemma.

\begin{lemma}
\label{w1_construction}
Given an instance $(H,G,\lambda)$ of \PartSubIso{} where $H$ is $3$-regular and has $k$ vertices,
one can in polynomial time construct an instance $\instanceSetCover$ of $\WeightedSegmentSetCover$ and a positive real $W$
such that:
\begin{enumerate}[label=(\arabic*)]
\item all segments in $\sets$ are axis-parallel;
\item \label{part1} if the instance $(H,G,\lambda)$ has a solution, then there exists a solution
to $\instanceSetCover$ of cardinality $\frac{11}{2}k$ and weight at most $W$; and
\item \label{part2} if there exists a solution to $\instanceSetCover$ of weight at most $W$,
then the instance $(H,G,\lambda)$ has a solution.
\end{enumerate}
\end{lemma}
Note that in \ref{part2} we in fact do not require any bound on the cardinality of the solution, just on its weight.

It is easy to see that Lemma~\ref{w1_construction} implies Theorem~\ref{thm:subisoHardness}. First, Lemma~\ref{w1_construction} gives a parameterized reduction from the $\mathsf{W}[1]$-hard \PartSubIso{} problem with $3$-regular pattern graphs to $\WeightedSegmentSetCover$ parameterized by the cardinality of the sought solution, which shows that the latter problem is $\mathsf{W}[1]$-hard as well. Second, combining the reduction with an algorithm for $\WeightedSegmentSetCover$ with running time as postulated in Theorem~\ref{w1_hard} would give an algorithm for \PartSubIso{} with running time $f(k)\cdot |V(G)|^{o(k/\log k)}$ for a computable function $f$, which would contradict ETH by Theorem~\ref{thm:subisoHardness}. So we are left with giving a proof of Lemma~\ref{w1_construction}, which spans the remainder of this section.

The key element of the proof will be a construction of a {\em{choice gadget}} that works within a single line; this construction is presented in the lemma below. Here, a {\em{chain}} is a sequence $(A_1,A_2,\ldots,A_\ell)$ of subsets of $\mathbb{R}$ such that for each $i\in \{1,\ldots,\ell-1\}$, all numbers in $A_i$ are strictly smaller than all numbers in $A_{i+1}$.

\begin{lemma}\label{lem:choice-gadget}
 Suppose we are given an integer $N>100$ and $p$ chains $\{(A_{j,1},\ldots,A_{j,\ell})\colon j\in \{1,\ldots,p\}\}$ of length $\ell$ each such that the sets $\{A_{j,t}\colon j\in \{1,\ldots,p\}, t\in \{1,\ldots,\ell\}\}$ are all pairwise disjoint and contained in $\{1,\ldots,N\}$. Then one can in polynomial time construct a set of points $\points\subseteq \mathbb{R}$, $\points\supseteq \{1,\ldots,N\}$, as well as a set of segments $\sets$ contained in $\mathbb{R}$ such that the following holds:
 \begin{itemize}
  \item For every $j\in \{1,\ldots,p\}$ and every set $B$ that contains exactly one point from each element of the chain $(A_{j,1},\ldots,A_{j,\ell})$, there exists $\Rr_B\subseteq \sets$ such that $|\Rr_B|=\ell+1$, $\Rr_B$ covers all points of $\points$ except for $B$, and the total length of the segments in $\Rr_B$ is equal to $N+1-2\ell/N^2$.
  \item For every subset of segments $\Rr\subseteq \sets$, if $\Rr$ covers all points in $\points\setminus \{1,\ldots,N\}$, then the total length of segments in $\Rr$ is at least $N+1-2/N$.
  \item For every subset of segments $\Rr\subseteq \sets$, if the total length of segments of $\Rr$ is not larger than $N+\frac{3}{2}$ and $\Rr$ covers all points in $\points\setminus \{1,\ldots,N\}$, then the total length of segments of $\Rr$ is equal to $N+1-2\ell/N^2$ and there exists $j\in \{1,\ldots,p\}$ such that for every $t\in \{1,\ldots,\ell\}$, $\Rr$ does not cover the whole set $A_{j,t}$.
 \end{itemize}
\end{lemma}
\begin{proof}
 Denote $I\coloneqq \{1,\ldots,N\}$ and $\epsilon\coloneqq 1/N^2$ for convenience. For every $i\in I$, let
 $$i^-\coloneqq i-\epsilon\qquad\textrm{and}\qquad i^+\coloneqq i+\epsilon.$$
 Define $I^-\coloneqq \{i^-\colon i\in I\}$, $I^+\coloneqq \{i^+\colon i\in I\}$, and
 $$\points\coloneqq \{0\}\cup I^-\cup I\cup I^+.$$
 Next, for every $j\in \{1,\ldots,p\}$, define the following set of segments:
 $$\Rr_j\coloneqq \{[0,a^-]\colon a\in A_{j,1}\}\cup \bigcup_{t=1}^{\ell-1} \{[a^+,b^-]\colon (a,b)\in A_{j,t}\times A_{j,t+1}\} \cup \{[a^+,N+1]\colon a\in A_{j,\ell}\}.$$
 We set
 $$\sets\coloneqq \bigcup_{j=1}^p \Rr_j.$$
 See Figure~\ref{fig:w1_choice} for a visualization of the construction.
 We are left with verifying the three postulated properties of $\points$ and $\sets$.

\input{figures/fig_w1_choice.tex}

 For the first property, let $b_t$ be the unique element of $B\cap A_{j,t}$, for $t\in \{1,\ldots,\ell\}$, and let
 $$\Rr_B\coloneqq \{[0,b_1^-],[b_1^+,b_2^-],\ldots,[b_{\ell-1}^+,b_\ell^-],[b_\ell^+,N+1]\}.$$
 It is straightforward to see that $\Rr_B$ covers all the points of $\points$ except for $B$, and that the total sum of lengths of segments in $\Rr_B$ is $N+1-2\ell\epsilon=N+1-2\ell/N^2$.

 For the second postulated property, observe that each segment of $\sets$ that covers any point $i^+\in I^+$, in fact covers the whole interval $[i^+,(i+1)^-]$ (where $(N+1)^-=N+1$). Similarly, each segment of $\sets$ that covers any point $i^-\in I^-$, in fact covers the whole interval $[(i-1)^+,i^-]$ (where $0^+=0$). Hence, if $\Rr\subseteq \sets$ covers all points of $\points \setminus I$, in particular $\Rr$ covers all points in $I^+\cup I^-$, hence also all intervals of the form $[i^+,(i+1)^-]$ for $i\in \{0,1,\ldots,N\}$. The sum of the lengths of those intevals is equal to $N+1-2\epsilon N=N+1-2/N$. Hence, the sum of length of intervals in $\Rr$ must be at least $N+1-2/N$.

 For the third postulated property, observe that if two segments of $\sets$ intersect, then their intersection is a segment of length at least $1-2\epsilon$. Since $\Rr$ covers all points of $\points\setminus I$, by the second property the sum of lengths of the segments in $\Rr$ is at least $N+1-2/N$. Now if any of those segments intersected, then the total sum of lengths of the segments in $\Rr$ would be at least $N+1-2/N+(1-2\epsilon)$, which is larger than $N+\frac{3}{2}$. We conclude that the segments of $\Rr$ are pairwise disjoint.

 Since $0\in \points\setminus I$, there is a segment $s_1\in \Rr$ that covers $0$. By construction, there exists $j\in \{1,\ldots,p\}$ such that $s_1=[0,b_{j,1}^-]$ for some $b_{j,1}\in A_{j,1}$. As the segments of $\Rr$ are pairwise disjoint and cover all points in $I^+$, the next (in the natural order on $\mathbb{R}$) segment in $\Rr$ must start at $b_{j,1}^+$, and in particular $b_{j,1}$ is not covered by $\Rr$. Since all sets in all chains on input are pairwise disjoint, the segment in $\Rr$ starting at $b_{j,1}^+$ must be of the form $s_2=[b_{j,1}^+,b_{j,2}^-]$ for some $b_{j,2}\in A_{j,2}$. Continuing this reasoning, we find that in fact $\Rr=\Rr_B$ for some set $B=\{b_{j,1},b_{j,2},\ldots,b_{j,\ell}\}$ such that $b_{j,t}\in A_{j,t}$ for each $t\in \{1,\ldots,\ell\}$. In particular, the total length of segments in $\Rr$ is equal to $N+1-2\epsilon\ell$ and $\Rr$ does not cover any point in $B$; the latter implies that for each $t\in \{1,\ldots,\ell\}$, $\Rr$ does not cover $A_{j,t}$ entirely.
\end{proof}

 \newcommand{\Dd}{\mathcal{D}}

With Lemma~\ref{lem:choice-gadget} established, we proceed to the proof of Lemma~\ref{w1_construction}.

 Let $(H,G,\lambda)$ be the given instance of \PartSubIso{} where $H$ is a $3$-regular graph. Let $k\coloneqq |V(H)|$ and $\ell\coloneqq |E(H)|$; note that $\ell=\frac{3}{2}k$. We may assume that $V(H)=\{1,\ldots,k\}$, and that whenever $uv$ is an edge in $G$, we have that $\lambda(u)\lambda(v)$ is an edge of $H$ (other edges in $G$ play no role in the problem and can be discarded).
 We construct an instance $\instanceSetCover$ of $\WeightedSegmentSetCover$ as follows; see Figure~\ref{fig_w1_whole} for a visualization.

 \input{figures/fig_w1_whole.tex}

 For each edge $ab\in E(H)$, let $E_{ab}$ be the subset of those edges $uv$ of $G$ for which $\lambda(u)=a$ and $\lambda(v)=b$. Thus, $\{E_{ab}\colon ab\in E(H)\}$ is a partition of $E(G)$.
 Let $N\coloneqq |E(G)|$ and $\xi\colon E(G)\to \{1,\ldots,N\}$ be any bijection such that for each $ab\in E(H)$, $\xi(E_{ab})$ is a contiguous interval of integers. By copying some vertices of $G$ if necessary, we may assume that $N>100k$.

 Consider any $a\in \{1,\ldots,k\}$ and let $b_1,b_2,b_3$ be the three neighbors of $a$ in $H$, ordered so that $(\xi(E_{ab_1}),\xi(E_{ab_2}),\xi(E_{ab_3}))$ is a chain. For each $u\in \lambda^{-1}(a)$, let $E_u$ be the set of edges of $G$ incident to $u$, and let us construct the chain
 $$C_u\coloneqq (\xi(E_u\cap E_{ab_1}),\xi(E_u\cap E_{ab_2}),\xi(E_u\cap E_{ab_3})).$$
 Note that all sets featured in all the chains $C_u$, for $u\in \lambda^{-1}(a)$, are pairwise disjoint.
 We now apply Lemma~\ref{lem:choice-gadget} for the integer $N$ and the chains $\{C_u\colon u\in \lambda^{-1}(a)\}$. This way, we construct a suitable point set $\points_a\subseteq \mathbb{R}$ and a set of segments $\sets_a$ contained in $\mathbb{R}$. We put all those points and segments on the line $\{(x,a)\colon x\in \mathbb{R}\}$; that is, every point $x\in \points_a$ is replaced with the point $(x,a)$, and similarly for the segments of $\sets_a$. By somehow abusing the notation, we let $\points_a$ and $\sets_a$ be the point set and the segment set after the~replacement.

 Next, for every edge $uv$ of $G$, we define $s_{uv}$ to be the segment with endpoints $(\xi(uv),a)$ and $(\xi(uv),b)$, where $a=\lambda(u)$ and $b=\lambda(v)$.

 We set
 $$\points\coloneqq \bigcup_{a=1}^k \points_a\qquad \textrm{and}\qquad \sets\coloneqq \{s_{uv}\colon uv\in E(G)\}\cup \bigcup_{a=1}^k \sets_a.$$
 Note that all segments in sets $\sets_a$ are horizontal and each segment $s_{uv}$ is vertical, thus $\sets$ consists of axis-parallel segments.
 Each segment $s\in \bigcup_{a=1}^k \sets_a$ is assigned weight $\mathbf{w}(s)$ equal to the length of $s$, and each segment $s_{uv}$ for $uv\in E(G)$ is assigned weight $\mathbf{w}(s_{uv})=\delta$, where $\delta\coloneqq 1/N^4$. Finally, we set
 $$W\coloneqq k\cdot (N+1-6/N^2)+\delta\ell.$$

 This concludes the construction of the instance $\instanceSetCover$. We are left with verifying the correctness of the reduction, which is done in the following two claims.

 \begin{claim}\label{cl:there}
  Suppose the input instance $(H,G,\lambda)$ of \PartSubIso{} has a solution. Then the output instance $\instanceSetCover$ of $\WeightedSegmentSetCover$ has a solution of cardinality $4k+\ell=\frac{11}{2}k$ and weight at most $W$.
 \end{claim}
 \begin{claimproof}
  Let $\phi$ be the supposed solution to $(H,G,\lambda)$. By the first property of Lemma~\ref{lem:choice-gadget}, for every $a\in \{1,\ldots,k\}$ there is a set $\Rr_{\phi,a}$ of size $4$ and total weight $N+1-6/N^2$ that covers all points from $\points_a$ except for the points $$(\xi(\phi(a)\phi(b_1)),a),(\xi(\phi(a)\phi(b_2)),a),(\xi(\phi(a)\phi(b_3)),a),$$
  where $b_1,b_2,b_3$ are the neighbors of $a$ in $H$.
  Define
  $$\sol\coloneqq \{s_{\phi(a)\phi(b)}\colon ab\in E(H)\}\cup \bigcup_{a=1}^k \Rr_{\phi,a}.$$
  Thus, for each $a\in \{1,\ldots,k\}$, the aforementioned points of $\points_a$ not covered by $\Rr_{\phi,a}$ are actually covered by the segments $s_{\phi(a)\phi(b_1)},s_{\phi(a)\phi(b_2)},s_{\phi(a)\phi(b_3)}$. We conclude that $\sol$ covers all the points in $\points$ and has cardinality $4k+\ell=\frac{11}{2}k$ and total weight $W$, as promised.
 \end{claimproof}

 \begin{claim}\label{cl:back}
  Suppose the output instance $\instanceSetCover$ of $\WeightedSegmentSetCover$ has a solution of weight at most $W$. Then the input instance $(H,G,\lambda)$ of \PartSubIso{} has a solution.
 \end{claim}
 \begin{claimproof}
  Let $\sol$ be the supposed solution to $\instanceSetCover$. Denote
  $$\Dd\coloneqq \sol\cap \{s_{uv}\colon uv\in E(G)\}$$
  and
  $$\sol_a\coloneqq \sol\cap \sets_a\qquad\textrm{for }a\in \{1,\ldots,k\}.$$
  Fix $a\in \{1,\ldots,k\}$ for a moment. Observe that the segments from $\Dd$ can only cover points with positive integer coordinates within the set $\points_a$, hence the whole point set $\points_a\setminus (\{1,\ldots,N\}\times \{a\})$ has to be covered by $\sol_a$. By the second property of Lemma~\ref{lem:choice-gadget} we infer that the total weight of $\sol_a$ must be at least $N+1-2/N$.

  Observe now that
  $$W-k\cdot (N+1-2/N)=\delta \ell+2k/N-6k/N^2<\frac{1}{2}.$$
  It follows that the total weight of each set $\sol_a$ must be smaller than $N+\frac{3}{2}$, for otherwise the sum of weights of sets $\sol_a$ would be larger than $W$. By the third property of Lemma~\ref{lem:choice-gadget}, we infer that for every $a\in \{1,\ldots,k\}$, the total weight of $\sol_a$ is equal to $N+1-6/N^2$ and there exists $\phi(a)\in \lambda^{-1}(a)$ such that $\sol_a$ does not entirely cover any of the sets $\xi(E_{\phi(a)}\cap E_{ab_1}),\xi(E_{\phi(a)}\cap E_{ab_2}),\xi(E_{\phi(a)}\cap E_{ab_3})$, where $b_1,b_2,b_3$ are the three neighbors of $a$ in $H$. In particular, there are edges $e_{a,b_1}\in E_{ab_1},e_{a,b_2}\in E_{ab_2},e_{a,b_3}\in E_{ab_3}$, all sharing the endpoint~$\phi(a)$, such that $\sol_a$ does not cover the points $(\xi(e_{a,b_1}),a),(\xi(e_{a,b_2}),a),(\xi(e_{a,b_3}),a)$. Call these points $X_a$ and let $X\coloneqq \bigcup_{a=1}^k X_a$. Note that
  $$|X|=3k=2\ell$$
  and that $X$ must be entirely covered by $\Dd$.

  Since the weight of $\sol_a$ is equal to $N+1-6/N^2$ for each $a\in \{1,\ldots,k\}$, the weight of $\Dd$ is upper bounded by
  $$W-k\cdot (N+1-6/N^2)=\delta \ell.$$
  As every member of $\Dd$ has weight $\delta$, we conclude that $|\Dd|\leq \ell$. Now, one can readily verify that every segment $s_{uv}\in \Dd$ can cover at most two points in $X$, as $X$ cannot contain more than two points with the same horizontal coordinate (recall that this coordinate is the index of an edge of $G$). Moreover, $s_{uv}$ can cover two points in $X$ only if $u=\phi(a)$ and $v=\phi(b)$, where $a=\lambda^{-1}(u)$ and $b=\lambda^{-1}(v)$. As $|X|=2\ell$ and $|\Dd|\leq \ell$, this must be the case for every segment in $\Dd$. In particular, $\phi(a)\phi(b)$ must be an edge in $G$ for every edge $ab\in E(H)$, so $\phi$ is a solution to the instance $(H,G,\lambda)$ of \PartSubIso.
 \end{claimproof}

 Claims~\ref{cl:there} and~\ref{cl:back} finish the proof of Lemma~\ref{w1_construction}. So the proof of Theorem~\ref{w1_hard} is also done.

 \section{Hardness of approximating $\SegmentSetCover$}\label{section:apx_hard}

 \newcommand{\setCoverInstance}{(\points, \sets)}
\newcommand{\approximate}{\mathsf{approx}^{*}}
\newcommand{\OurSAT}{\textsf{Max-(E3,E5)-SAT}\xspace}

In this section we prove our hardness of approximation result for $\SegmentSetCover$ --- Theorem~\ref{segment_cover_apx_hard} --- restated below for convenience.

\segmentCoverApxHard*

The proof of Theorem~\ref{segment_cover_apx_hard} relies on a gap-preserving reduction from the \OurSAT problem, for which a hardness of approximation statement has been established by Arora et al.~\cite{hastad}. We first need to recall these results.

\subsection{\OurSAT}

A CNF formula $\varphi$ is called an {\em{(E3,E5)-formula}} if every clause of $\varphi$ contains exactly $3$ literals and every variable appears in $\varphi$ exactly $5$ times. Note that by double-counting variable-clause incidences, an (E3,E5)-formula with $n$ variables has exactly $\frac{5}{3}n$ clauses, and in particular $n$ is divisible by $3$.
Then \OurSAT is the optimization problem defined as follows. On input, we are given an (E3,E5)-formula $\varphi$, and the task is to find a boolean assignment to the variables of $\varphi$ that satisfies the maximum possible number of clauses.

In~\cite{AroraLMSS98}, Arora et al. proved the following statement, from which we will reduce. Here, a CNF formula $\varphi$ is called {\em{$\alpha$-satisfiable}}, for $\alpha\in [0,1]$, if there exists a boolean assignment to the variables of $\varphi$ that satisfies at least an $\alpha$ fraction of all the clauses, and $\varphi$ is {\em{satisfiable}} if it is $1$-satisfiable.

\begin{theorem}[\cite{AroraLMSS98}, see also Theorem~2.24 in~\cite{hastad}]
	\label{hastadtheorem}
	There exists $\alpha<1$ such that it is $\mathsf{NP}$-hard to distinguish satisfiable
	(E3,E5)-formulas from (E3,E5)-formulas that are
	less than~$\alpha$-satisfiable.
\end{theorem}

\subsection{Statement of the reduction}

As mentioned, our proof of Theorem~\ref{segment_cover_apx_hard} relies on a gap-preserving reduction from \OurSAT to $\SegmentSetCover$. The statement of this reduction is encapsulated in the following lemma.

\begin{lemma}
	\label{apxconstruction}
	Given an instance $\varphi$ of \OurSAT
	with $n$ variables and $m$ clauses,
	one can construct an instance $\setCoverInstance$ of $\SegmentSetCover$ with
	axis-parallel segments satisfying the following properties:
	\begin{enumerate}[label={(\arabic*)},nosep]
	\item \label{item:apxconstruction_correctness}
	If $\varphi$ has a solution that satisfies $m-k$ clauses,
	then there exists a solution to the instance $\setCoverInstance$ of cardinality at most $\frac{64}{3}n+k$.

	\item \label{item:apxconstruction_completness}
	If the instance $\setCoverInstance$ has a solution of cardinality $\frac{64}{3}n+k$, then
	there exists a solution to $\varphi$ that satisfies at least  $m-5k$ clauses.

	\item \label{lemma:apxconstruction:enumerate:extension}
	For any $\sol \subseteq \sets$, if $\sol^{+\frac{1}{2}}$
	covers all the points in $\points$, then $\sol$ also covers all the points in~$\points$.
\end{enumerate}
\end{lemma}

We prove Lemma~\ref{apxconstruction} in
subsequent sections. Section~\ref{construction_description}
describes the constructed instance $\setCoverInstance$.
Property~\ref{item:apxconstruction_correctness} is verified in Lemma~\ref{construction_correctness},
\ref{item:apxconstruction_completness} in Lemma~\ref{construction_completness},
and finally \ref{lemma:apxconstruction:enumerate:extension} trivially
follows from Lemma~\ref{lemma:exntension_robust}.
But before we proceed to the proof, let us show how
Theorem~\ref{segment_cover_apx_hard} follows from combining Lemma~\ref{apxconstruction}
and Theorem~\ref{hastadtheorem}.

\begin{proof}[Proof of Theorem \ref{segment_cover_apx_hard}]
We set
$$\gamma\coloneqq \frac{1-\alpha}{64},$$
where $\alpha<1$ is the constant provided by Theorem~\ref{hastadtheorem}. Towards contradiction, suppose that there is a polynomial-time algorithm that given an instance $(\points,\sets)$ of $\SegmentSetCover$ with axis-parallel segments, outputs $\sol\subseteq \sets$ such that $\sol^{+\frac{1}{2}}$ covers $\points$ and $|\sol|\leq (1+\gamma)\opt$, where $\opt$ is the minimum cardinality of a subset of $\sets$ that covers $\points$.

Consider the following polynomial-time algorithm for the \OurSAT problem: given an (E3,E5)-formula $\varphi$ with $n$ variables and $m$ clauses,
\begin{itemize}[nosep]
 \item apply the algorithm of Lemma~\ref{apxconstruction} to compute an instance $\setCoverInstance$ of $\SegmentSetCover$;
 \item apply the assumed approximation algorithm for $\SegmentSetCover$ with $\frac{1}{2}$-extension, yielding a suitable subfamily $\sol\subseteq \sets$; and
 \item output that $\varphi$ is satisfiable if $|\sol|\leq (1+\gamma)\cdot \frac{64}{3}n$, and otherwise output that $\varphi$ is less than $\alpha$-satisfiable.
\end{itemize}
We claim that this algorithm correctly distinguishes satisfiable (E3,E5)-formulas from those that are less than $\alpha$-satisfiable. As the algorithm works in polynomial time, by Theorem~\ref{hastadtheorem} this would entail that $\mathsf{P}=\mathsf{NP}$.

First, suppose $\varphi$ is satisfiable. By Lemma~\ref{apxconstruction}, property~\ref{item:apxconstruction_correctness}, instance $(\points,\sets)$ has a solution of cardinality at most $\frac{64}{3}n$. Hence, the computed set $\sol$ will have cardinality at most $(1+\gamma)\cdot\frac{64}{3}n$, and the algorithm will correctly conclude that $\varphi$ is satisfiable.

Second, suppose $\varphi$ is less than $\alpha$-satisfiable. Letting $\opt$ be the minimum size of a subset of $\sets$ that covers $\points$, from Lemma~\ref{apxconstruction}, property~\ref{item:apxconstruction_completness}, we conclude that
$$\opt>\frac{64}{3}\cdot n +\frac{(1-\alpha)}{5}\cdot m=\frac{64}{3}\cdot n+\frac{(1-\alpha)}{3}\cdot n=\frac{65-\alpha}{3}\cdot n=(1+\gamma)\cdot \frac{64}{3}n.$$
By assumption, the set $\sol\subseteq \sets$ computed by the algorithm is such that $\sol^{+\frac{1}{2}}$ covers $\points$. By Lemma~\ref{apxconstruction}, property~\ref{lemma:apxconstruction:enumerate:extension}, in fact we even have that $\sol$ covers $\points$, and hence it must hold that $|\sol|\geq \opt$. We conclude that $|\sol|>(1+\gamma)\cdot \frac{64}{3}n$ and the algorithm will correctly conclude that $\varphi$ is less than $\alpha$-satisfiable.
\end{proof}

\subsection{Construction}
\label{construction_description}
We proceed to the proof of Lemma \ref{apxconstruction}.
Let then $\varphi$ be the input instance of \OurSAT, let $x_1,\ldots,x_n$ be the variables of $\varphi$, and let $C_1,\ldots,C_m$ be the clauses of $\varphi$. Recall that since $\varphi$ is an (E3,E5)-formula, we have $m=\frac{5}{3}n$.

The construction will be composed of two types of gadgets:
{\em{variable-gadgets}} and {\em{clause-gadgets}}.
Clause-gadgets will be constructed using two {\em{OR-gadgets}}
connected together.

\subsubsection{Variable-gadgets}

A variable-gadget is responsible for choosing the value of a variable of $\varphi$. It allows two minimum solutions of size 3 each.
These two choices correspond to the two possible boolean values of the variable
corresponding to the gadget.

\paragraph*{Points.} Define points $a,b,c,d,e,f,g,h$ as follows, where $L \coloneqq 100n$:

\newcommand{\pointsVarNoArg}{\mathsf{pointsVariable} }
\newcommand{\pointsVar}[1]{\mathsf{pointsVariable}_{#1} }
\newcommand{\chooseVar}[2]{\mathsf{chooseVariable}^{#1}_{#2} }
\newcommand{\segmentsVar}[1]{\mathsf{segmentsVariable}_{#1} }

\input{figures/fig_apx_choose_variable.tex}

\begin{center}
\begin{tabular}{ l l l l}
	$a \coloneqq (-3L, 0)$ &
	$b \coloneqq (-2L, 0)$ &
	$c \coloneqq (-L, 0)$ &
	$d \coloneqq (-3L, 1)$ \\
	$e \coloneqq (-2L, 1)$ &
	$f \coloneqq (-2L, 2)$ &
	$g \coloneqq (L, 0)$ &
	$h \coloneqq (L, 2)$
\end{tabular}
\end{center}

Let us define:
$$\pointsVarNoArg \coloneqq  \{a, b, c, d, e, f\}$$
and, for every $i\in \{1,\ldots,n\}$,
$$\pointsVar{i} \coloneqq \pointsVarNoArg + (0, 4i).$$
(That is, $\pointsVar{i}$ consists of points of $\pointsVarNoArg$ translated by vector $(0,4i)$.) We denote $a_i \coloneqq a + (0,4i)$ etc.

\paragraph*{Segments.}
Let us define segments:

\newcommand{\xTrueSegmentDef}[1]{(c_{#1}, g_{#1})}
\newcommand{\xFalseSegmentDef}[1]{(f_{#1}, h_{#1})}
\newcommand{\xTrueSegment}[1]{\mathsf{xTrueSegment}_{#1}}
\newcommand{\xFalseSegment}[1]{\mathsf{xFalseSegment}_{#1}}
\newcommand{\orTrueSegment}[2]{(t_{#1, #2}, v_{#1, #2})}

$$\chooseVar{\true}{i} \coloneqq \{ (a_i, d_i), (b_i, f_i), (c_i, g_i)\},$$
$$\chooseVar{\false}{i} \coloneqq \{(a_i, c_i), (d_i, e_i), (f_i, h_i)\},$$
$$\segmentsVar{i} \coloneqq \chooseVar{\true}{i} \cup \chooseVar{\false}{i}.$$
We also name two of these segments for future reference:
$$\xTrueSegment{i} \coloneqq \xTrueSegmentDef{i}\qquad\textrm{and}\qquad \xFalseSegment{i} \coloneqq \xFalseSegmentDef{i}.$$

The next three simple lemmas verify the properties of the variable-gadget.

\begin{lemma}
\label{choose_variables_solution}
For every $i\in \{1,\ldots,n\}$, points in $\pointsVar{i}$
can be covered using 3 segments from $\segmentsVar{i}$.
\end{lemma}

\begin{proof}
We can use either set $\chooseVar{\true}{i}$ or $\chooseVar{\false}{i}$.
\end{proof}

\begin{lemma}
\label{choose_variables_no_less}
For every $i\in \{1,\ldots,n\}$, points in $\pointsVar{i}$
cannot be covered with fewer than 3 segments from $\segmentsVar{i}$.
\end{lemma}

\begin{proof}
No segment of $\segmentsVar{i}$ covers more than one point from
$\{d_i, f_i, c_i\}$, so $\pointsVar{i}$ cannot be covered with fewer than 3 segments.
\end{proof}

\begin{lemma}
\label{choose_variables_both}
For every set $A \subseteq \segmentsVar{i}$ such that $\xTrueSegment{i}, \xFalseSegment{i} \in A$ and $A$ covers $\pointsVar{i}$,
it holds that $|A| \ge 4$.
\end{lemma}
\begin{proof}
No segment from $\segmentsVar{i}$ covers more than one point from
$\{a_i, e_i\}$,
so
$\pointsVar{i}\setminus \{c_i, f_i\}$
cannot be covered with fewer than 2 segments.
\end{proof}

\subsubsection{OR-gadget}

An OR-gadget connects input and output segments (see Figure \ref{fig:apx_or_gadget})
in a way that is supposed to simulate  binary disjunction.
Input segments are the only segments that cover points outside of the gadget,
as their left endpoints lie outside of it.
Point $v_{i,j}$ is the only one that can be covered
by segments that do not belong to the gadget.

The OR-gadget has the property that every set of segments
that covers all the points in the gadget uses at least 3 segments from it.
Moreover, the output segment belongs to the solution of size 3
only if at least one of the input segments belongs to the solution.
Therefore, optimum solutions restricted to the OR-gadget behave
like the binary disjunction of the input segments.

\paragraph*{Points.}
We define

\newcommand{\chooseOr}[3]{\mathsf{chooseOr}^{#1}_{#2,#3}}
\newcommand{\orMoveVariable}[2]{\mathsf{orMoveVariable}_{#1,#2}}
\newcommand{\pointsOr}[2]{\mathsf{pointsOr}_{#1,#2}}
\newcommand{\segmentsOr}[2]{\mathsf{segmentsOr}_{#1,#2}}
\newcommand{\vc}{\xi}

\input{figures/fig_apx_or_gadget.tex}

\newcommand{\clauseXFactor}{20}

\begin{center}
	\begin{tabular}{ l l l l}
		$l \coloneqq (0, 0)$, &
		$m \coloneqq (0, 1)$, &
		$n \coloneqq (0, 2)$, &
		$o \coloneqq (0, 3)$, \\
		$p \coloneqq (0, 4)$, &
		$q \coloneqq (1, 1)$, &
		$r \coloneqq (1, 3)$, &
		$s \coloneqq (2, 1)$, \\
		$t \coloneqq (2, 2)$, &
		$u \coloneqq (2, 3)$, &
		$v \coloneqq (3, 2)$, &
	\end{tabular}
\end{center}
as well as
$$\vc_{i, j} \coloneqq (\clauseXFactor i + 3 + 3j, 4(n+1) + 2j).$$
Similarly as before, for integers $i,j$ we
define
$\{ l_{i, j}, m_{i, j}, \ldots, v_{i, j} \}$
as $\{l, m, \ldots, v\}$ shifted by~$\vc_{i, j}$;
that is, ${l_{i,j} = l + \vc_{i,j}}$ etc.

Note that $v_{i, 0} = l_{i, 1}$, see Figure~\ref{fig:apx_clause}.
Next, we let
$$\pointsOr{i}{j} \coloneqq
 \{l_{i, j}, m_{i, j}, n_{i, j}, o_{i, j},
 p_{i, j}, q_{i, j}, r_{i, j}, s_{i, j}, t_{i, j}, u_{i, j} \}.
 $$
Note that $\pointsOr{i}{j}$ does {\em{not}} include the point $v_{i,j}$.

\paragraph*{Segments.}
The segment set in an OR-gadget is defined in several parts:
\begin{eqnarray*}
\chooseOr{\false}{i}{j} &\coloneqq& \{ (q_{i, j}, r_{i, j}), (s_{i, j}, u_{i, j})\}, \\
\chooseOr{\true}{i}{j} &\coloneqq& \{ (m_{i, j}, s_{i, j}), (o_{i, j}, u_{i, j}), (t_{i, j}, v_{i, j}) \}, \\
\orMoveVariable{i}{j} &\coloneqq& \{ (l_{i, j}, n_{i, j}), (n_{i, j}, p_{i, j})\}.
\end{eqnarray*}
Finally, all segments of an OR-gadget are defined as:
$$\segmentsOr{i}{j} \coloneqq
  \chooseOr{\false}{i}{j} \cup \chooseOr{\true}{i}{j} \cup \orMoveVariable{i}{j}.
$$

Again, the following lemmas verify the properties of the OR-gadgets.

\begin{lemma}
\label{cover_or_true}
For all $i\in \{1,\ldots,m\}$, $j \in \{0, 1\}$, and
 $x \in \{l_{i, j}, p_{i, j}\}$, all the points in
${\pointsOr_{i, j} \setminus \{ x\} \cup \{v_{i, j}\}}$
can be covered
with 4 segments from $\segmentsOr{i}{j}$.
\end{lemma}

\begin{proof}
We can do this using one segment from
$\orMoveVariable{i}{j}$, the one that does not cover~$x$,
and all segments from $\chooseOr{\true}{i}{j}$.
\end{proof}

\begin{lemma}
\label{cover_or_false}
For all $i\in \{1,\ldots,m\}$ and $j \in \{0, 1\}$, points in
$\pointsOr{i}{j}$ can be covered
with 4 segments from $\segmentsOr{i}{j}$.
\end{lemma}

\begin{proof}
We can do this using segments from $\orMoveVariable{i}{j} \cup \chooseOr{\false}{i}{j}$.
\end{proof}

\subsubsection{Clause-gadget}
A clause-gadget is responsible for determining whether the boolean values encoded in variable gadgets
satisfy the corresponding clause in the input formula $\varphi$.
It has a minimum solution of size $11$
if and only if the clause is satisfied, i.e. at least one
of the respective variables is assigned the correct value.
Otherwise, its minimum solution has size $12$.
In this way, by analyzing the size of the minimum
solution to the entire constructed instance, we will be able to tell
how many clauses it is possible to satisfy
in an optimum solution to $\varphi$.

\newcommand{\pointsClause}[1]{\mathsf{pointsClause}_{#1}}
\newcommand{\negate}{\mathsf{neg}}
\newcommand{\idx}{\mathsf{idx}}

\input{figures/fig_apx_clause.tex}

\paragraph*{Points.}
First, we define auxiliary functions for literals.
For a literal $w$, let $\idx(w)\in \{1,\ldots,n\}$ be the index of the variable in $w$,
and $\negate(w)$ be the Boolean value (0 or 1) indicating whether the variable in $w$ is negated
or not.
\begin{eqnarray*}
\idx(w) & \coloneqq &  i \text{ when } w = x_i \\
\negate(w) & \coloneqq &
\begin{cases}
 0 & \text{if } w = x_i \\
 1 & \text{if } w = \neg x_i
\end{cases}
\end{eqnarray*}

Consider a clause $C_i = a \lor b \lor c$, where $a,b,c$ are literals. Then, we define points in the gadget corresponding to $C_i$ as follows:

\begin{center}
\begin{tabular}{ l l }
	$x_{i, 0} \coloneqq (\clauseXFactor i, 4\cdot \idx(a) + 2\cdot \negate(c)),$ &
	$x_{i, 1} \coloneqq (\clauseXFactor i, 4(n+1)),$ \\
	$y_{i, 0} \coloneqq (\clauseXFactor i+1, 4\cdot \idx(b) + 2\cdot \negate(b)),$ &
	$y_{i, 1} \coloneqq (\clauseXFactor i+1, 4(n+1) + 4),$ \\
	$z_{i, 0} \coloneqq (\clauseXFactor i+2, 4\cdot \idx(c) + 2\cdot \negate(c)),$ &
	$z_{i, 1} \coloneqq (\clauseXFactor i+2, 4(n+1) + 6).$
\end{tabular}
\end{center}

\newcommand{\segmentsClause}{\mathsf{segmentsClause}}
\newcommand{\moveVariablePoints}[1]{\mathsf{moveVariablePoints}_{#1}}
\newcommand{\moveVariableSegments}[1]{\mathsf{moveVariableSegments}_{#1}}

We are now ready to define the set of points in a clause-gadget:
 $$\moveVariablePoints{i} \coloneqq
 \{x_{i, j} \colon j \in \{0, 1\}\} \cup
 \{y_{i, j} \colon j \in \{0, 1\}\} \cup
 \{z_{i, j} \colon j \in \{0, 1\}\},
 $$
 $$\pointsClause{i} \coloneqq
 \moveVariablePoints{i} \cup \pointsOr{i}{0}
 \cup \pointsOr{i}{1} \cup \{v_{i, 1} \}.
 $$
Note that the following two points are equal: $v_{i,0} = l_{i,1}$.
This translates to the fact that the output of the first OR-gadget
is an input to the second OR-gadget.
Intuitively, this creates a disjunction of 3 boolean~values.

\paragraph*{Segments.}
We also define segments for the clause-gadget as below:
\begin{eqnarray*}
\moveVariableSegments{i} & \coloneqq & \{
(x_{i, 0}, x_{i, 1}),
(y_{i, 0}, y_{i, 1}),
(z_{i, 0}, z_{i, 1}),
(x_{i, 1}, l_{i, 0}),
(y_{i, 1}, p_{i, 0}),
(z_{i, 1}, p_{i, 1})
\} \\
\segmentsClause_i & \coloneqq & \moveVariableSegments{i} \cup \segmentsOr{i}{0} \cup \segmentsOr{i}{1}.
\end{eqnarray*}

\newcommand{\segmentsClauseSolTrue}[1]{\mathsf{solClause}^{\true,#1}}
\newcommand{\segmentsClauseSolFalse}{\mathsf{solClause}^{\false}}

The clause-gadgets consist of two OR-gadgets.
Ideally, we would place the $i$-th clause-gadget close to the
$\xTrueSegment{j_1}$ or $\xFalseSegment{j_1}$ segments
corresponding to the literals that occur in the $i$-th clause.
It would be inconvenient to position them there,
because between these segments there may be additional
$\xTrueSegment{j_2}$ or $\xFalseSegment{j_2}$
segments corresponding to the other literals.

Instead, we use simple auxiliary gadgets to
\textit{transfer} the value whether a segment
is in a solution.
Each transfer gadget consists of two segments $(x_{i, 0}, x_{i, 1}), (x_{i, 1}, a)$ (and similarly for $y$ and $z$).
These are the only segments that can cover $x_{i,1}$.
We place $x_{i,0}$ on a segment that we want to transfer (i.e.
segment responsible for choosing the variable value satisfying the
corresponding literal).
If in some solution $x_{i,0}$ is already covered by this segment, then
we can cover $x_{i,1}$ by $(x_{i,1}, a)$, thus also covering $a$.
If $x_{i,0}$ is not covered by this segment,
then the only way to cover $x_{i,0}$ is to use segment $(x_{i, 0}, x_{i, 1})$.
Intuitively,
in any optimal solution the two segments \textit{transfer} the state of whether $x_{i,0}$
is covered onto whether $a$ is covered.
Therefore, the number of segments in the optimal solution is increased by one,
and we get a point $a$ that was effectively placed
on some segment $s$, but it can be placed anywhere in the plane instead,
consequently simplifying the construction.

We now formally verify properties of the construction.

\begin{lemma}
\label{cover_clauses_solution_true}
For all $i\in \{1,\ldots,m\}$ and $a \in \{ x_{i, 0}, y_{i, 0}, z_{i, 0}\}$,
there is a set $\segmentsClauseSolTrue{a}_i \subseteq \segmentsClause_i$
with $|\segmentsClauseSolTrue{a}_i| = 11$
that covers all points in $\pointsClause{i} \setminus \{a\}$.
\end{lemma}

\begin{proof}
For $a = x_{i, 0}$ (analogous proof for $y_{i, 0}$):
First use Lemma~\ref{cover_or_true} twice with excluded $x = l_{i, 0}$ and
$x = l_{i, 1} = v_{i, 0}$,
resulting with 8 segments in $\chooseOr{\true}{i}{0} \cup \chooseOr{\true}{i}{1}$
which cover all the required points apart from
$x_{i, 1}, y_{i, 0}, y_{i, 1}, z_{i, 0}, z_{i, 1}, l_{i, 0}$.
We cover those using additional 3 segments:
$(x_{i, 1}, l_{i, 0}), (y_{i, 0}, y_{i, 1}),
(z_{i, 0}, z_{i, 1})$.

For $a = z_{0, i}$:
Using Lemma~\ref{cover_or_false} and Lemma~\ref{cover_or_true} with
$x = p_{i, 1}$,
we obtain 8 segments in $\chooseOr{\false}{i}{0} \cup \chooseOr{\true}{i}{1}$
which cover all required points apart from
$x_{i, 0}, x_{i, 1}, y_{i, 0}, y_{i, 1}, z_{i, 1}, p_{i, 1}$.
We cover those using additional 3 segments:
$\{ (x_{i, 0}, x_{i, 1}), (y_{i, 0}, y_{i, 1}),
(z_{i, 1}, p_{i, 1}) \}$.
\end{proof}

\begin{lemma}
\label{cover_clauses_solution_false}
For every $i\in \{1,\ldots,m\}$, there is
a set $\segmentsClauseSolFalse_i \subseteq \segmentsClause_i$
of size $12$
that covers all points in $\pointsClause{i}$.
\end{lemma}

\begin{proof}
Using Lemma \ref{cover_or_false} twice we
cover $\pointsOr{i}{0}$ and  $\pointsOr{i}{1}$
with 8 segments.
To cover the remaining points we additionally use:
$\{ (x_{i, 0}, x_{i, 1}), (y_{i, 0}, y_{i, 1}),
(z_{i, 0}, z_{i, 1}), (t_{i, 1}, v_{i, 1}) \}.$
\end{proof}

\begin{lemma}
\label{cover_clauses_segments_no_less}
For every $i\in \{1,\ldots,m\}$,
\begin{enumerate}[label={(\arabic*)},nosep]
	\item points in $\pointsClause{i}$ cannot be covered
	using any subset of segments
	from $\segmentsClause_i$ of size smaller than $12$; and
	\item points in $\pointsClause{i} \setminus \{ x_{i, 0}, y_{i, 0}, z_{i, 0}\}$
	cannot be covered using any subset of segments
	from $\segmentsClause_i$ of size smaller than $11$.
\end{enumerate}
\end{lemma}

\begin{proof}[Proof of (1).]
No segment in $\segmentsClause_i$ covers more than one point from
$$\{ x_{i, 0}, y_{i, 0}, z_{i, 0}, l_{i, 0}, p_{i, 0}, q_{i, 0},
u_{i, 0}, v_{i, 0} = l_{i, 1}, p_{i, 1}, q_{i, 1}, u_{i, 1}, v_{i, 1} \}.$$
Therefore, we need to use at least $12$ segments.
\end{proof}

\begin{proof}[Proof of (2).]
We can define disjoint sets $X, Y, Z$ such that
$$X \cup Y \cup Z \subseteq \pointsClause{i} \setminus \{x_{i, 0}, y_{i, 0}, z_{i, 0}\}$$
and there are no segments in $\segmentsClause_i$ covering points from different sets.
And we prove a lower bound for each of these sets.

First, let
\begin{equation*}X \coloneqq \{x_{i, 1}, y_{i, 1}, z_{i, 1}\}.\end{equation*}
No two points in $X$ can be covered with one segment
of $\segmentsClause_i$, so it must be covered with 3 different segments.
Next, we define the other sets:
\begin{eqnarray*}
Y & \coloneqq & \pointsOr{i}{0} \setminus \{l_{i, 0}, p_{i, 0}\},\\
Z & \coloneqq & \pointsOr{i}{1} \setminus \{l_{i, 1}, p_{i, 1}\}.
\end{eqnarray*}
For both $Y$ and $Z$ we can check all of the subsets of 3 segments
of $\segmentsClause_i$
to conclude that none of them cover the considered points,
so both $Y$ and $Z$ have to be covered with
disjoint sets of 4 segments~each.

Therefore, $\pointsClause{i} \setminus \{x_{i, 0}, y_{i, 0}, z_{i, 0}\}$
must be covered with at least $3 + 4 + 4 = 11$ segments from $\segmentsClause_i$.
\end{proof}

\subsubsection{Summary}

\input{figures/fig_apx_segment_whole.tex}

Finally we define the set of points and segments for the constructed instance:
$$\points \coloneqq \bigcup_{i=1}^n \pointsVar{i} \cup \bigcup_{i=1}^{m}\pointsClause{i},$$
$$\sets \coloneqq \bigcup_{i=1}^{n} \segmentsVar{i} \cup \bigcup_{i=1}^{m}\segmentsClause_i.$$
The next observation immediately implies property~\ref{lemma:apxconstruction:enumerate:extension}.
\begin{lemma}[Robustness to $\frac{1}{2}$-extension]
\label{lemma:exntension_robust}
For every segment $s \in \sets$,
$s$ and $s^{+\frac{1}{2}}$ cover the same points from~$\points$.
\end{lemma}
\begin{proof}
 A straightforward verification of all types of segments used in the construction, we leave the details to the reader. (Recall here that the segment $s^{+\frac{1}{2}}$ is considered without endpoints, this detail is necessary in a few checks.)
\end{proof}

\subsection{Correctness}

We are left with arguing the correctness of the construction, that is, properties \ref{item:apxconstruction_correctness} and \ref{item:apxconstruction_completness}.
Property \ref{item:apxconstruction_correctness} is proved in the following lemma.

\begin{lemma}
	\label{construction_correctness}
	Suppose there is a boolean assignment $\eta \colon \{ x_1, x_2, \ldots, x_n\} \rightarrow \{\true, \false\}$ that satisfies $m-k$ clauses in $\varphi$. Then
	the instance $\setCoverInstance$ of $\SegmentSetCover$
	admits a solution of cardinality $\frac{64}{3}n + k$.
\end{lemma}
\begin{proof}
We cover every variable-gadget with solution described in
Lemma~\ref{choose_variables_solution}, where
in the $i$-th gadget we choose the set of segments corresponding to the
value of $\eta(x_i)$.

For every clause that is satisfied, say $C_i$, let $\ell_i$ be the literal satisfies $C_i$, and let the point corresponding to $\ell_i$ in $\pointsClause{i}$ be $a$ (this is a point of the form $x_{i,0}$, $y_{i,0}$, or $z_{i,0}$).
Points in $\pointsClause{i}$
can be covered with set $\segmentsClauseSolTrue{a}_i$ described in
Lemma~\ref{cover_clauses_solution_true}.
For every clause that is not satisfied, say $C_j$,
points in $\pointsClause{j}$ can be covered
with the set $\segmentsClauseSolFalse_j$ described in
Lemma~\ref{cover_clauses_solution_false}.

\newcommand{\Qq}{\mathcal{Q}}

Formally, we define
sets $\Rr_i$ and $\Qq_i$ responsible for covering variable and clause gadgets as follows:
\begin{align*}
	\Rr_i & \coloneqq \begin{cases}
		\chooseVar{\true}{i} & \text{if}\ \eta(x_i) = \true, \\
		\chooseVar{\false}{i} & \text{if}\ \eta(x_i) = \false. \\
		\end{cases} \\
	\Qq_i & \coloneqq \begin{cases}
		\segmentsClauseSolTrue{a}_i & \text{if}\ C_i \text{ satisfied in }\eta\text{ by the literal corresponding to } a, \\
		\segmentsClauseSolFalse_i & \text{if}\ C_i \text{ not satisfied in }\eta.
		\end{cases}
\end{align*}
Then, we define
\begin{align*}
\sol & \coloneqq \bigcup_{i=1}^{n} \Rr_i \cup \bigcup_{i=1}^m \Qq_i.
\end{align*}
It is then straightforward to verify using respective lemmas provided in Section~\ref{construction_description} that $\sol$ constructed as above covers all the points from $\points$.
All of these sets $\Rr_i$ and $\Qq_i$ are disjoint, so the size of the obtained solution is:
$$|\sol| = \sum_{i=1}^{n} |\Rr_i| + \sum_{i=1}^m |\Qq_i| = 3n + 11(m-k) + 12k = \frac{64}{3}n + k.$$
In the last equality, we used the fact that $m=\frac{5}{3}n$.
\end{proof}

\begin{lemma}	\label{construction_completness}
	Suppose there is a solution $\sol$ of the instance $\setCoverInstance$
	of $\SegmentSetCover$ of cardinality $\frac{64}{3}n+k$.
	Then there is a boolean assignment $\eta\colon \{x_1,\ldots,x_n\}\to \{\true,\false\}$
	that satisfies at least $m-5k$ clauses.
\end{lemma}

\begin{proof}
Call a variable $x_i$, $i\in \{1,\ldots,n\}$, {\em{overpaid}} if $|\segmentsVar{i} \cap \sol| \ge 4$. Similarly, call a clause $C_i$, $i\in \{1,\ldots,m\}$, {\em{overpaid}} if $|\segmentsClause_{i} \cap \sol| \ge 12$. Let $p$ and $q$ be the  number of overpaid variables and the number of overpaid clauses, respectively. By Lemma~\ref{choose_variables_no_less} we have $|\segmentsVar{i} \cap \sol| \ge 3$ for every $i\in \{1,\ldots,n\}$, and by Lemma~\ref{cover_clauses_segments_no_less} we have $|\segmentsClause_i \cap \sol| \ge 11$ for all $i\in \{1,\ldots,m\}$. Hence, we have
$$\frac{64}{3}n+k=|\sol|=\geq 3n+11m+p+q=\frac{64}{3}n+p+q,$$
implying that
$$k\geq p+q.$$

We define the following assignment $\eta$. For $i\in \{1,\ldots,n\}$, if $x_i$ is not overpaid then
\begin{align*}
\eta(x_i) \coloneqq & \begin{cases}
	\true & \text{if}\ \xTrueSegment{i} \in \sol, \\
	\false & \text{otherwise}
	\end{cases}
\end{align*}
Otherwise, if $x_i$ is overpaid, then we simply set
\begin{align*}
\eta(x_i) \coloneqq & \true.
\end{align*}

Let us now bound the number of clauses satisfied by $\eta$. Consider a clause $C_i$, $i\in \{1,\ldots,m\}$. Note that if $C_i$ is not satisfied in $\eta$ and none of the variables involved in $C_i$ is overpaid, then none of the points $x_{i,0},y_{i,0},z_{i,0}$ is covered by segments belonging to variable gadget. By Lemma~\ref{cover_clauses_segments_no_less}, this implies that at least one of the following assertions must hold for every $C_i$ that is not satisfied in $\eta$:
\begin{itemize}[nosep]
 \item $C_i$ is overpaid, or
 \item one of the variables involved in $C_i$ is overpaid.
\end{itemize}
Since every variable appears in exactly $5$ clauses, it follows that the total number of clauses not satisfied in $\eta$ is at most $5p+q$. Put differently, the assignment $\eta$ satisfies at least $m-5p-q$ clauses. As $5p+q\leq 5k$, we conclude that $\eta$ satisfies at least $m-5k$ clauses, as promised.
\end{proof}

Now Lemma~\ref{apxconstruction} follows immediately from
Lemmas~\ref{lemma:exntension_robust}, \ref{construction_correctness}, and
\ref{construction_completness}.

 \paragraph*{Acknowledgements.} The authors would like to thank Krzysztof Maziarz for help with proofreading the manuscript.

\bibliography{bibl}

\end{document}

%% file: fpt_v_t_def.pdf_tex
\begingroup%
  \makeatletter%
  \providecommand\color[2][]{%
    \errmessage{(Inkscape) Color is used for the text in Inkscape, but the package 'color.sty' is not loaded}%
    \renewcommand\color[2][]{}%
  }%
  \providecommand\transparent[1]{%
    \errmessage{(Inkscape) Transparency is used (non-zero) for the text in Inkscape, but the package 'transparent.sty' is not loaded}%
    \renewcommand\transparent[1]{}%
  }%
  \providecommand\rotatebox[2]{#2}%
  \newcommand*\fsize{\dimexpr\f@size pt\relax}%
  \newcommand*\lineheight[1]{\fontsize{\fsize}{#1\fsize}\selectfont}%
  \ifx\svgwidth\undefined%
    \setlength{\unitlength}{124.64601902bp}%
    \ifx\svgscale\undefined%
      \relax%
    \else%
      \setlength{\unitlength}{\unitlength * \real{\svgscale}}%
    \fi%
  \else%
    \setlength{\unitlength}{\svgwidth}%
  \fi%
  \global\let\svgwidth\undefined%
  \global\let\svgscale\undefined%
  \makeatother%
  \begin{picture}(1,0.3089624)%
    \lineheight{1}%
    \setlength\tabcolsep{0pt}%
    \put(0,0){\includegraphics[width=\unitlength,page=1]{fpt_v_t_def.pdf}}%
    \put(0.26449754,5.37334612){\color[rgb]{0,0,0}\makebox(0,0)[lt]{\begin{minipage}{0.91861755\unitlength}\end{minipage}}}%
    \put(0.0373069,0.28212087){\color[rgb]{0,0,0}\makebox(0,0)[lt]{\lineheight{1.25}\smash{\begin{tabular}[t]{l}$v_1$\\\end{tabular}}}}%
    \put(0.1721398,0.28430715){\color[rgb]{0,0,0}\makebox(0,0)[lt]{\lineheight{1.25}\smash{\begin{tabular}[t]{l}$v_2$\\\end{tabular}}}}%
    \put(0.30729697,0.28283159){\color[rgb]{0,0,0}\makebox(0,0)[lt]{\lineheight{1.25}\smash{\begin{tabular}[t]{l}$v_3$\\\end{tabular}}}}%
    \put(0.44535451,0.27618606){\color[rgb]{0,0,0}\makebox(0,0)[lt]{\lineheight{1.25}\smash{\begin{tabular}[t]{l}$v_4$\\\end{tabular}}}}%
    \put(0.58225189,0.27618645){\color[rgb]{0,0,0}\makebox(0,0)[lt]{\lineheight{1.25}\smash{\begin{tabular}[t]{l}$v_5$\\\end{tabular}}}}%
    \put(0.72262975,0.2802468){\color[rgb]{0,0,0}\makebox(0,0)[lt]{\lineheight{1.25}\smash{\begin{tabular}[t]{l}$v_6$\\\end{tabular}}}}%
    \put(0.85314631,0.28140701){\color[rgb]{0,0,0}\makebox(0,0)[lt]{\lineheight{1.25}\smash{\begin{tabular}[t]{l}$v_7$\\\end{tabular}}}}%
    \put(0,0){\includegraphics[width=\unitlength,page=2]{fpt_v_t_def.pdf}}%
    \put(0.03730685,0.16340314){\color[rgb]{0,0,0}\makebox(0,0)[lt]{\lineheight{1.25}\smash{\begin{tabular}[t]{l}$t_1$\\\end{tabular}}}}%
    \put(0.17213978,0.16558942){\color[rgb]{0,0,0}\makebox(0,0)[lt]{\lineheight{1.25}\smash{\begin{tabular}[t]{l}$t_2$\\\end{tabular}}}}%
    \put(0.44289343,0.16157025){\color[rgb]{0,0,0}\makebox(0,0)[lt]{\lineheight{1.25}\smash{\begin{tabular}[t]{l}$t_4$\\\end{tabular}}}}%
    \put(0.75900317,0.16152907){\color[rgb]{0,0,0}\makebox(0,0)[lt]{\lineheight{1.25}\smash{\begin{tabular}[t]{l}$t_6$\\\end{tabular}}}}%
    \put(0.00051549,0.0937684){\color[rgb]{0,0,0}\makebox(0,0)[lt]{\begin{minipage}{0.04252785\unitlength}\raggedright $a$\end{minipage}}}%
    \put(0.13866061,0.09789902){\color[rgb]{0,0,0}\makebox(0,0)[lt]{\begin{minipage}{0.0359531\unitlength}\raggedright $b$\end{minipage}}}%
    \put(0.39161503,0.09789624){\color[rgb]{0,0,0}\makebox(0,0)[lt]{\begin{minipage}{0.12448511\unitlength}\raggedright $c=t_3$ \end{minipage}}}%
    \put(0.93809485,0.09767885){\color[rgb]{0,0,0}\makebox(0,0)[lt]{\begin{minipage}{0.08653129\unitlength}\raggedright $d=t_7$\end{minipage}}}%
    \put(-0.00357384,0.21060902){\color[rgb]{0,0,0}\makebox(0,0)[lt]{\begin{minipage}{0.093719\unitlength}\raggedright $a$\end{minipage}}}%
    \put(0.13432055,0.21060902){\color[rgb]{0,0,0}\makebox(0,0)[lt]{\begin{minipage}{0.093719\unitlength}\raggedright $b$\end{minipage}}}%
    \put(0.41132709,0.21492469){\color[rgb]{0,0,0}\makebox(0,0)[lt]{\begin{minipage}{0.32833856\unitlength}\raggedright $c$ \end{minipage}}}%
    \put(0.96008824,0.21083265){\color[rgb]{0,0,0}\makebox(0,0)[lt]{\begin{minipage}{0.08212713\unitlength}\raggedright $d$\end{minipage}}}%
    \put(0,0){\includegraphics[width=\unitlength,page=3]{fpt_v_t_def.pdf}}%
    \put(0.17311215,0.05867319){\color[rgb]{0,0,0}\makebox(0,0)[lt]{\begin{minipage}{0.07938995\unitlength}\raggedright $y$\end{minipage}}}%
    \put(0.72482188,0.05789682){\color[rgb]{0,0,0}\makebox(0,0)[lt]{\begin{minipage}{0.0948699\unitlength}\raggedright $z$\end{minipage}}}%
  \end{picture}%
\endgroup%

%% file: figures/fig_w1_choice.tex
\definecolor{environment}{RGB}{100,100,100}
\definecolor{move_variable1}{RGB}{137, 218, 109}
\definecolor{move_variable2}{RGB}{255, 207, 60}
\definecolor{choose_true1}{RGB}{56, 209, 241}
\definecolor{choose_true2}{RGB}{35, 140, 210}
\definecolor{choose_false}{RGB}{255, 9, 9}

\begin{figure}[h]
\centering
\begin{tikzpicture}[scale=1.2]

   \tikzstyle{point}=[circle,fill=black,minimum size=0.15cm,inner sep=0pt]
   \tikzstyle{guard}=[rectangle,draw=black!40,fill=black!20,minimum size=0.12cm,inner sep=0pt]
   \tikzstyle{cover}=[circle,draw=black!40,fill=black!40,minimum size=0.15cm,inner sep=0pt]

   \draw[thick] (-1,0) -- (10,0);
   \draw[line width=3, blue] (0,0) -- (3-0.2,0);
   \draw[line width=3, blue] (3+0.2,0) -- (7-0.2,0);
   \draw[line width=3, blue] (7+0.2,0) -- (9,0);
   
   \foreach \i in {1,2,3,4,5,6,7,8} {
     \node[cover] (p\i) at (\i,0) {};
     \draw[below] (p\i) node {{\scriptsize{$\i$}}};
     \node[guard] at (\i-0.2,0) {};
     \node[guard] at (\i+0.2,0) {};
   }
   
   \node[cover] (q) at (0,0) {};
   \draw[below] (q) node {{\scriptsize{$0$}}};
   \draw[below] (9,0) node {{\scriptsize{$9$}}};
   
  \node[point] at (p3) {};
  \node[point] at (p7) {};

\end{tikzpicture}
\caption{Construction of Lemma~\ref{lem:choice-gadget} for $N=8$. Elements of $I\cup \{0\}$ are depicted with circles and elements of $I^+\cup I^-$ are depicted with squares. Blue segments represent the set $\Rr_B$ for $B=\{3,7\}$.
}
\label{fig:w1_choice}
\end{figure}
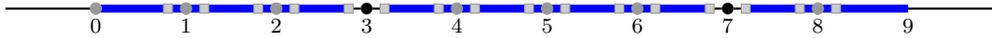

%% file: figures/fig_w1_whole.tex
\definecolor{environment}{RGB}{100,100,100}
\definecolor{move_variable1}{RGB}{137, 218, 109}
\definecolor{move_variable2}{RGB}{255, 207, 60}
\definecolor{choose_true1}{RGB}{56, 209, 241}
\definecolor{choose_true2}{RGB}{35, 140, 210}
\definecolor{choose_false}{RGB}{255, 9, 9}

\begin{figure}[h]
\centering
\begin{tikzpicture}[scale=1.2]

   \tikzstyle{point}=[circle,fill=black,minimum size=0.2cm,inner sep=0pt]
   \tikzstyle{guard}=[rectangle,draw=black!40,fill=black!20,minimum size=0.12cm,inner sep=0pt]
   \tikzstyle{cover}=[circle,draw=black!40,fill=black!40,minimum size=0.15cm,inner sep=0pt]

   \foreach \i in {0,1,2,3} {
   \draw[thick,black!50] (-1,\i*0.5) -- (8,\i*0.5);
     
   }
   
   \foreach \m/\i/\j/\k in {0/1/2/3,0.5/1/4/5,1/2/4/6,1.5/3/5/6} {
        \draw[line width=3, blue] (0,\m) -- (\i-0.2,\m);
        \draw[line width=3, blue] (\i+0.2,\m) -- (\j-0.2,\m);
        \draw[line width=3, blue] (\j+0.2,\m) -- (\k-0.2,\m);
        \draw[line width=3, blue] (\k+0.2,\m) -- (7,\m);
   }
   
   \foreach \i/\j/\k in {1/0/1,2/0/2,3/0/3,4/1/2,5/1/3,6/2/3} {
        \draw[line width=3,orange] (\i,\j*0.5) -- (\i,\k*0.5);
        \node[point] at (\i,\j*0.5) {};
        \node[point] at (\i,\k*0.5) {};
   }

\end{tikzpicture}
\caption{Example solution in the instance $(\points,\sets)$ constructed in the proof of Lemma~\ref{w1_construction} for $H=K_4$. Blue segments belong to the sets $\sol_i$ for $i\in \{1,2,3,4\}$ and orange segments belong to $\Dd$.
}
\label{fig_w1_whole}
\end{figure}
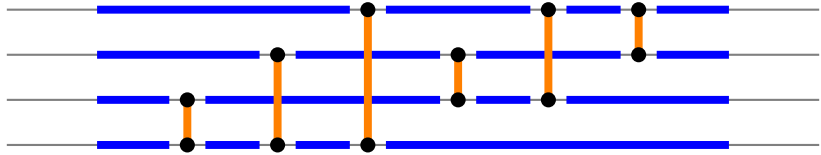

%% file: figures/fig_apx_choose_variable.tex
\definecolor{x_true_colour}{RGB}{40, 40, 255}
\definecolor{x_false_colour}{RGB}{255, 40, 40}

{\tikzset{point/.style={
    circle, draw=black, fill, fill=black, minimum size=4pt,inner sep=0pt, outer sep=0pt,
    prefix after command= {\pgfextra{\tikzset{every
    label/.style={label distance=0.05cm,text=black}}}}
    }
}

{\tikzset{point_not_cover/.style={
    circle, draw=black, fill, fill=white, minimum size=4pt,inner sep=0pt, outer sep=0pt,
    prefix after command= {\pgfextra{\tikzset{every
    label/.style={label distance=0.05cm,text=black}}}}
    }
}

\begin{figure}[h]
\centering
\begin{tikzpicture}
\tikzmath{
\stepx=1.5;
\stepy=1;
\y1=0;
\y2=\y1+\stepy;
\y3=\y2+\stepy;
\x1=0;
\x2=\x1+\stepx;
\x3=\x2+\stepx;
\xend=\x3+2*\stepx;
}

\draw[x_false_colour,very thick] (\x1,\y1) -- (\x3,\y1);
\draw[x_false_colour,very thick] (\x1,\y2) -- (\x2,\y2);
\draw[x_false_colour,very thick] (\x2,\y3) -- (\xend,\y3);
\draw[x_true_colour,very thick] (\x1,\y1) -- (\x1,\y2);
\draw[x_true_colour,very thick] (\x2,\y1) -- (\x2,\y3);
\draw[x_true_colour,very thick] (\x3,\y1) -- (\xend,\y1);

\node[point,label={below:$a_i$}] at (\x1,\y1) {};
\node[point,label={below:$b_i$}] at (\x2,\y1) {};
\node[point,label={below:$c_i$}] at (\x3,\y1) {};
\node[point,label={left:$d_i$}] at (\x1,\y2) {};
\node[point,label={above left:$e_i$}] at (\x2,\y2) {};
\node[point,label={above left:$f_i$}] at (\x2,\y3) {};
\node[point_not_cover,label={right:$g_i$}] at (\xend,\y1) {};
\node[point_not_cover,label={right:$h_i$}] at (\xend,\y3) {};

\end{tikzpicture}
\caption{Variable-gadget. We denote the set of points marked with black circles as $\pointsVar{i}$,
and they need to be covered (are part of the set $\points$).
Note that some of the points are not marked as black dots
and exists only to name segments for further reference.
We denote the set of \textcolor{x_false_colour}{red} segments as $\chooseVar{\false}{i}$
and the set of \textcolor{x_true_colour}{blue} segments as $\chooseVar{\true}{i}$.}
\label{fig:apx_choose_variable}
\end{figure}
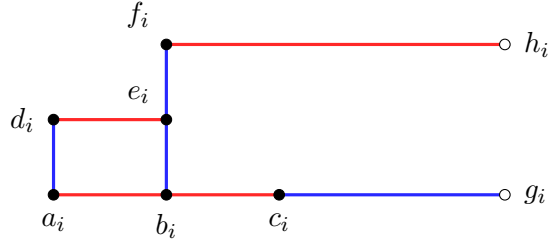

%% file: figures/fig_apx_or_gadget.tex
\definecolor{environment}{RGB}{100,100,100}
\definecolor{move_variable1}{RGB}{137, 218, 109}
\definecolor{move_variable2}{RGB}{255, 207, 60}
\definecolor{choose_true1}{RGB}{56, 209, 241}
\definecolor{choose_true2}{RGB}{35, 140, 210}
\definecolor{choose_false}{RGB}{255, 9, 9}

{\tikzset{point/.style={
    circle, draw=black, fill, fill=black, minimum size=4pt,inner sep=0pt, outer sep=0pt,
    prefix after command= {\pgfextra{\tikzset{every
    label/.style={label distance=0.05cm,text=black}}}}
    }
}

\begin{figure}[h]
\centering
\def\svgwidth{0.5\columnwidth}
\begin{tikzpicture}
\tikzmath{
\stepx=1.5;
\stepy=1.5;
\xbeg=0;
\x1=\xbeg+3*\stepx;
\x2=\x1+\stepx;
\x3=\x2+\stepx;
\xend=\x3+\stepx;
\y1=0;
\y2=\y1+\stepy;
\y3=\y2+\stepy;
\y4=\y3+\stepy;
\y5=\y4+\stepy;
}

\draw[environment,ultra thick] (\xbeg,\y1) -- (\x1,\y1) node[black,pos=0.15, above] {$input_x$};
\draw[environment,ultra thick] (\xbeg,\y5) -- (\x1,\y5) node[black,pos=0.15, above] {$input_y$};
\draw[move_variable1, ultra thick] (\x1,\y1) -- (\x1,\y3);
\draw[move_variable2, ultra thick] (\x1,\y5) -- (\x1,\y3);
\draw[choose_true1, ultra thick] (\x1,\y2) -- (\x3,\y2);
\draw[choose_true1, ultra thick] (\x1,\y4) -- (\x3,\y4);
\draw[choose_true2, ultra thick] (\x3,\y3) -- (\xend,\y3) node [black,pos=0.5, below] {$output$};
\draw[choose_false, ultra thick] (\x2,\y2) -- (\x2,\y4);
\draw[choose_false, ultra thick] (\x3,\y2) -- (\x3,\y4);

\node[point,label={above left:$l_{i,j}$}] at (\x1,\y1) {};
\node[point,label={left:$m_{i,j}$}] at (\x1,\y2) {};
\node[point,label={left:$n_{i,j}$}] at (\x1,\y3) {};
\node[point,label={left:$o_{i,j}$}] at (\x1,\y4) {};
\node[point,label={above left:$p_{i,j}$}] at (\x1,\y5) {};

\node[point,label={below:$q_{i,j}$}] at (\x2,\y2) {};
\node[point,label={above:$r_{i,j}$}] at (\x2,\y4) {};
\node[point,label={below:$s_{i,j}$}] at (\x3,\y2) {};
\node[point,label={left:$t_{i,j}$}] at (\x3,\y3) {};
\node[point,label={above:$u_{i,j}$}] at (\x3,\y4) {};
\node[point,label={above:$v_{i,j}$}] at (\xend,\y3) {};

\end{tikzpicture}
\caption{
	\textbf{OR-gadget.} Segments from $\chooseOr{\false}{i}{j}$ are \textcolor{choose_false}{red},
	segments from $\chooseOr{\true}{i}{j}$ are blue
	(both \textcolor{choose_true1}{light blue} and \textcolor{choose_true2}{dark blue}),
	segments from $\orMoveVariable{i}{j}$ are \textcolor{move_variable1}{green} and \textcolor{move_variable2}{yellow}.
	\textcolor{choose_true2}{Dark blue} segment is the $output$ segment.
	\textcolor{environment}{Grey segments} $input_x$ and $input_y$ are input segments that
	are not part of $\segmentsOr{i}{j}$.
}
\label{fig:apx_or_gadget}
\end{figure}
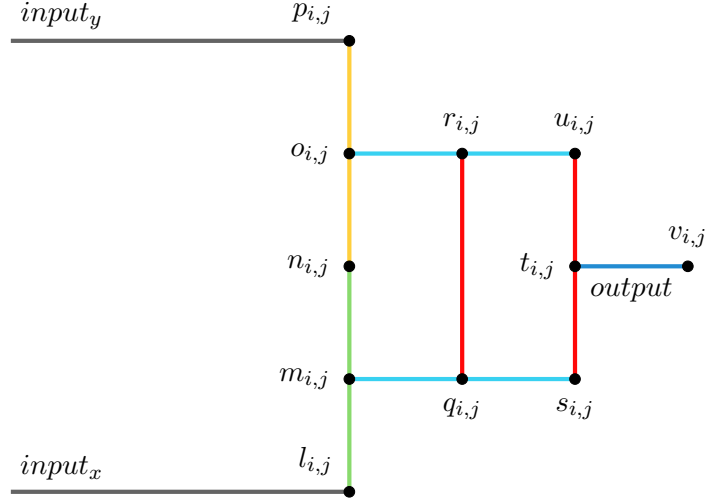

%% file: figures/fig_apx_clause.tex
\definecolor{environment}{RGB}{150,150,150}

{\tikzset{point/.style={
    circle, draw=black, fill, fill=black, minimum size=4pt,inner sep=0pt, outer sep=0pt,
    prefix after command= {\pgfextra{\tikzset{every
    label/.style={label distance=0.05cm,text=black,font=\footnotesize}}}}
    }
}

\begin{figure}[h]
\centering
\begin{tikzpicture}[scale=0.8]
\tikzmath{
\stepx=1.7;
\stepy=1.5;
\boxsize=1.2*\stepx;
\xbeg=0;
\x1=\xbeg+2*\stepx;
\x2=\x1+\stepx;
\x3=\x2+\stepx;
\x4=\x3+\stepx;
\x5=\x4+\boxsize;
\x6=\x5+\stepx;
\x7=\x6+\boxsize;
\x8=\x7+\stepx;
\y1=0;
\y2=\y1+\stepy;
\y3=\y2+\stepy;
\y4=\y3+2*\stepy;
\y5=\y4+0.5*\stepy;
\y6=\y5+0.5*\stepy;
\y7=\y6+0.5*\stepy;
}

\draw[thick] (\x1,\y3) -- (\x1,\y4);
\draw[thick] (\x2,\y1) -- (\x2,\y6);
\draw[thick] (\x3,\y2) -- (\x3,\y7);
\draw[thick] (\x1,\y4) -- (\x4,\y4);
\draw[thick] (\x2,\y6) -- (\x4,\y6);
\draw[thick] (\x3,\y7) -- (\x6,\y7);
\draw[thick] (\x5,\y5) -- (\x6,\y5);
\draw[thick] (\x7,\y6) -- (\x8,\y6);
\draw [fill=green!50, draw=green!50] (\x4,\y6) rectangle (\x5,\y4) node[pos=0.5] {\scriptsize OR-gadget};
\draw [fill=green!50, draw=green!50] (\x6,\y5) rectangle (\x7,\y7) node[pos=0.5] {\scriptsize OR-gadget};

\draw[environment,thick] (\xbeg, \y1) -- (\x8, \y1) node[black,pos=0.7, below] {$\xTrueSegment{b}$};
\draw[environment,thick] (\xbeg, \y2) -- (\x8, \y2) node[black,pos=0.7, below] {$\xFalseSegment{c}$};
\draw[environment,thick] (\xbeg, \y3) -- (\x8, \y3) node[black,pos=0.7, above] {$\xTrueSegment{a}$};

\node[point,label={above left:$x_{i,0}$}] at (\x1,\y3) {};
\node[point,label={above left:$y_{i,0}$}] at (\x2,\y1) {};
\node[point,label={above left:$z_{i,0}$}] at (\x3,\y2) {};
\node[point,label={above left:$x_{i,1}$}] at (\x1,\y4) {};
\node[point,label={above left:$y_{i,1}$}] at (\x2,\y6) {};
\node[point,label={above left:$z_{i,1}$}] at (\x3,\y7) {};
\node[point,label={above left:$l_{i,0}$}] at (\x4,\y4) {};
\node[point,label={above left:$p_{i,0}$}] at (\x4,\y6) {};
\node[point,label={above right:$t_{i,0}$}] at (\x5,\y5) {};
\node[point,label={below:$v_{i,0}=l_{i,1}$}] at (\x6,\y5) {};
\node[point,label={above:$p_{i,1}$}] at (\x6,\y7) {};
\node[point,label={above right:$t_{i,1}$}] at (\x7,\y6) {};
\node[point,label={above right:$v_{i,1}$}] at (\x8,\y6) {};
\end{tikzpicture}

\caption{Clause-gadget for a clause $a \lor b \lor \neg c$.
Every green rectangle is an OR-gadget.
$y$-coordinates of $x_{i, 0}$, $y_{i, 0}$ and $z_{i,0}$
depend on the variables in the $i$-th clause.
Grey segments corresponds to the values of variables
satisfying the $i$-th clause.
}
\label{fig:apx_clause}
\end{figure}
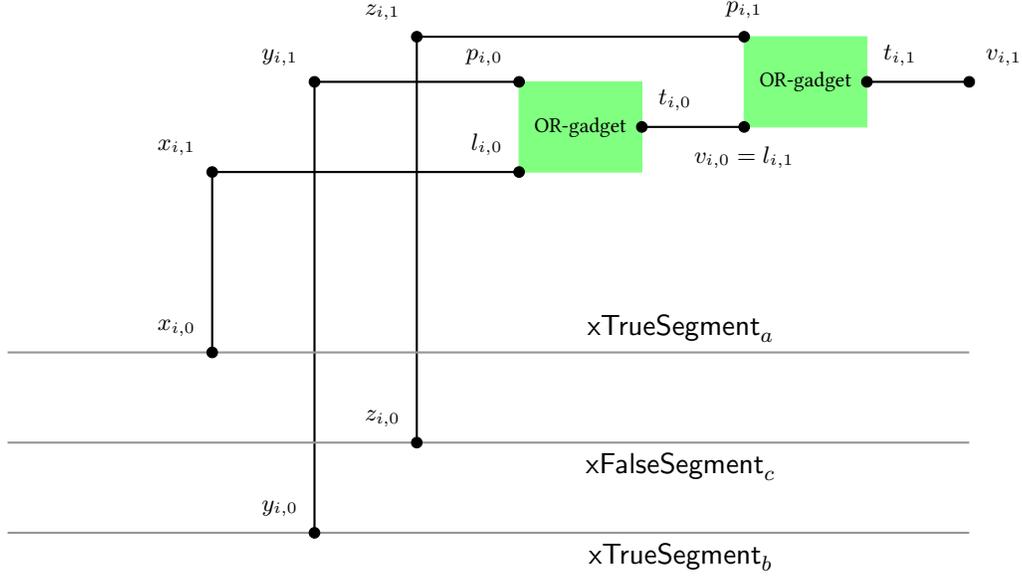

%% file: figures/fig_apx_segment_whole.tex
{

{\tikzset{node/.style={
    prefix after command= {\pgfextra{\tikzset{every
    label/.style={font=\footnotesize}}}}
    }
}
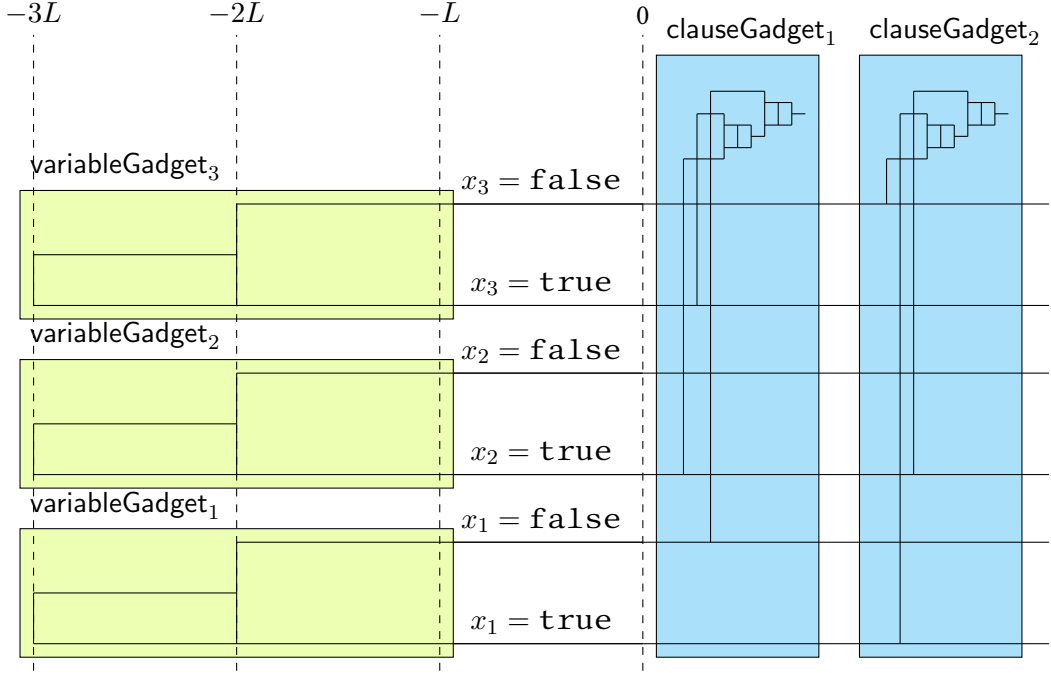
\begin{figure}
\centering
\begin{tikzpicture}[scale=0.9]
\tikzmath{
\x0=0;
\y0=0;
\yalabel=1.5;
\yblabel=2.5;
\yclabel=4.0;
\ydlabel=5.0;
\yelabel=6.5;
\yflabel=7.5;
\yglabel=9.0;
\mepsx0=-0.2;
\yhlabel=1.5;
\xalabel=9;
\xblabel=3;
\xclabel=6;
\yilabel=0.75;
\mepsy0=-0.2;
\epsxclabel=6.2;
\epsyhlabel=1.7;
\yajlabel=1.5;
\yaalabel=2.5;
\yablabel=1.5;
\yaclabel=2.5;
\yadlabel=4.0;
\xdlabel=9;
\xelabel=3;
\xflabel=6;
\yaelabel=3.25;
\mepsyaalabel=2.3;
\epsxflabel=6.2;
\epsyadlabel=4.2;
\yaflabel=1.5;
\yaglabel=2.5;
\yahlabel=4.0;
\yailabel=5.0;
\ybjlabel=1.5;
\ybalabel=2.5;
\ybblabel=4.0;
\ybclabel=5.0;
\ybdlabel=6.5;
\xglabel=9;
\xhlabel=3;
\xilabel=6;
\ybelabel=5.75;
\mepsyailabel=4.8;
\epsxilabel=6.2;
\epsybdlabel=6.7;
\xajlabel=9.4;
\xaalabel=9.4;
\xablabel=11.4;
\ybflabel=1.5;
\ybglabel=2.5;
\ybhlabel=4.0;
\ybilabel=5.0;
\ycjlabel=6.5;
\ycalabel=8.5;
\ycblabel=1.5;
\ycclabel=2.5;
\ycdlabel=4.0;
\ycelabel=5.0;
\ycflabel=6.5;
\xaclabel=9.6;
\ycglabel=7.166666666666667;
\xadlabel=9.8;
\ychlabel=7.833333333333333;
\xaelabel=10.0;
\ycilabel=8.166666666666666;
\xaflabel=10.200000000000001;
\xaglabel=10.8;
\mepsxajlabel=9.200000000000001;
\epsxablabel=11.6;
\epsycalabel=8.7;
\ydjlabel=1.5;
\ydalabel=2.5;
\ydblabel=1.5;
\ydclabel=2.5;
\yddlabel=4.0;
\ydelabel=5.0;
\ydflabel=1.5;
\ydglabel=7.833333333333334;
\ydhlabel=7.333333333333334;
\ydilabel=7.5;
\yejlabel=7.666666666666667;
\xahlabel=10.4;
\xailabel=10.6;
\xbjlabel=10.799999999999999;
\yealabel=7.499999999999999;
\yeblabel=8.166666666666666;
\yeclabel=7.666666666666666;
\yedlabel=7.833333333333332;
\yeelabel=7.999999999999999;
\xbalabel=11.0;
\xbblabel=11.2;
\xbclabel=11.399999999999999;
\xbdlabel=9.4;
\xbelabel=11.4;
\xbflabel=12.4;
\xbglabel=9.4;
\xbhlabel=11.4;
\xbilabel=12.4;
\xcjlabel=14.4;
\yeflabel=1.5;
\yeglabel=2.5;
\yehlabel=4.0;
\yeilabel=5.0;
\yfjlabel=6.5;
\yfalabel=8.5;
\yfblabel=1.5;
\yfclabel=2.5;
\yfdlabel=4.0;
\yfelabel=5.0;
\yfflabel=6.5;
\xcalabel=12.6;
\yfglabel=7.166666666666667;
\xcblabel=12.8;
\yfhlabel=7.833333333333333;
\xcclabel=13.0;
\yfilabel=8.166666666666666;
\xcdlabel=13.200000000000001;
\xcelabel=13.8;
\mepsxbflabel=12.200000000000001;
\epsxcjlabel=14.6;
\epsyfalabel=8.7;
\ygjlabel=1.5;
\ygalabel=2.5;
\ygblabel=4.0;
\ygclabel=5.0;
\ygdlabel=6.5;
\ygelabel=1.5;
\ygflabel=2.5;
\ygglabel=7.833333333333334;
\yghlabel=7.333333333333334;
\ygilabel=7.5;
\yhjlabel=7.666666666666667;
\xcflabel=13.4;
\xcglabel=13.6;
\xchlabel=13.799999999999999;
\yhalabel=7.499999999999999;
\yhblabel=8.166666666666666;
\yhclabel=7.666666666666666;
\yhdlabel=7.833333333333332;
\yhelabel=7.999999999999999;
\xcilabel=14.0;
\xdjlabel=14.2;
\xdalabel=14.399999999999999;
}

\filldraw [fill=cyan!30, draw=black] (\mepsxajlabel,\mepsy0) rectangle (\epsxablabel, \epsycalabel);
\draw (\xaclabel,\ydalabel) -- (\xaclabel,\ycglabel);
\draw (\xaclabel,\ycglabel) -- (\xaflabel,\ycglabel);
\draw (\xadlabel,\ydelabel) -- (\xadlabel,\ychlabel);
\draw (\xadlabel,\ychlabel) -- (\xaflabel,\ychlabel);
\draw (\xaelabel,\ydflabel) -- (\xaelabel,\ycilabel);
\draw (\xaelabel,\ycilabel) -- (\xaglabel,\ycilabel);
\node[above right] at (\mepsxajlabel,\epsycalabel) {$\mathsf{clauseGadget}_1$};
\draw (\xaflabel,\ycglabel) -- (\xaflabel,\ydglabel);
\draw (\xaflabel,\ydhlabel) -- (\xailabel,\ydhlabel);
\draw (\xaflabel,\yejlabel) -- (\xailabel,\yejlabel);
\draw (\xahlabel,\ydhlabel) -- (\xahlabel,\yejlabel);
\draw (\xailabel,\ydhlabel) -- (\xailabel,\yejlabel);
\draw (\xailabel,\ydilabel) -- (\xbjlabel,\ydilabel);
\draw (\xaglabel,\yealabel) -- (\xaglabel,\yeblabel);
\draw (\xaglabel,\yeclabel) -- (\xbblabel,\yeclabel);
\draw (\xaglabel,\yeelabel) -- (\xbblabel,\yeelabel);
\draw (\xbalabel,\yeclabel) -- (\xbalabel,\yeelabel);
\draw (\xbblabel,\yeclabel) -- (\xbblabel,\yeelabel);
\draw (\xbblabel,\yedlabel) -- (\xbclabel,\yedlabel);

\filldraw [fill=cyan!30, draw=black] (\mepsxbflabel,\mepsy0) rectangle (\epsxcjlabel, \epsyfalabel);
\draw (\xcalabel,\ygdlabel) -- (\xcalabel,\yfglabel);
\draw (\xcalabel,\yfglabel) -- (\xcdlabel,\yfglabel);
\draw (\xcblabel,\y0) -- (\xcblabel,\yfhlabel);
\draw (\xcblabel,\yfhlabel) -- (\xcdlabel,\yfhlabel);
\draw (\xcclabel,\ygflabel) -- (\xcclabel,\yfilabel);
\draw (\xcclabel,\yfilabel) -- (\xcelabel,\yfilabel);
\node[above right] at (\mepsxbflabel,\epsyfalabel) {$\mathsf{clauseGadget}_2$};
\draw (\xcdlabel,\yfglabel) -- (\xcdlabel,\ygglabel);
\draw (\xcdlabel,\yghlabel) -- (\xcglabel,\yghlabel);
\draw (\xcdlabel,\yhjlabel) -- (\xcglabel,\yhjlabel);
\draw (\xcflabel,\yghlabel) -- (\xcflabel,\yhjlabel);
\draw (\xcglabel,\yghlabel) -- (\xcglabel,\yhjlabel);
\draw (\xcglabel,\ygilabel) -- (\xchlabel,\ygilabel);
\draw (\xcelabel,\yhalabel) -- (\xcelabel,\yhblabel);
\draw (\xcelabel,\yhclabel) -- (\xdjlabel,\yhclabel);
\draw (\xcelabel,\yhelabel) -- (\xdjlabel,\yhelabel);
\draw (\xcilabel,\yhclabel) -- (\xcilabel,\yhelabel);
\draw (\xdjlabel,\yhclabel) -- (\xdjlabel,\yhelabel);
\draw (\xdjlabel,\yhdlabel) -- (\xdalabel,\yhdlabel);
\draw (\xclabel,\y0) -- (\xalabel,\y0) node[pos=0.5, above] {$x_1 = \true$};
\draw (\xclabel,\yhlabel) -- (\xalabel,\yhlabel) node[pos=0.5, above] {$x_1 = \false$};
\draw (\xalabel,\y0) -- (15,\y0);
\draw (\xalabel,\yhlabel) -- (15,\yhlabel);
\filldraw [fill=lime!30, draw=black] (\mepsx0,\mepsy0) rectangle (\epsxclabel, \epsyhlabel);
\draw (\x0,\y0) -- (\xalabel,\y0);
\draw (\x0,\y0) -- (\x0,\yilabel);
\draw (\xblabel,\y0) -- (\xblabel,\yhlabel);
\draw (\xblabel,\yhlabel) -- (\xalabel,\yhlabel);
\draw (\x0,\yilabel) -- (\xblabel,\yilabel);
\node[above right] at (\mepsx0,\epsyhlabel) {$\mathsf{variableGadget}_1$};
\draw (\xflabel,\yaalabel) -- (\xdlabel,\yaalabel) node[pos=0.5, above] {$x_2 = \true$};
\draw (\xflabel,\yadlabel) -- (\xdlabel,\yadlabel) node[pos=0.5, above] {$x_2 = \false$};
\draw (\xdlabel,\yaalabel) -- (15,\yaalabel);
\draw (\xdlabel,\yadlabel) -- (15,\yadlabel);
\filldraw [fill=lime!30, draw=black] (\mepsx0,\mepsyaalabel) rectangle (\epsxflabel, \epsyadlabel);
\draw (\x0,\yaalabel) -- (\xdlabel,\yaalabel);
\draw (\x0,\yaalabel) -- (\x0,\yaelabel);
\draw (\xelabel,\yaalabel) -- (\xelabel,\yadlabel);
\draw (\xelabel,\yadlabel) -- (\xdlabel,\yadlabel);
\draw (\x0,\yaelabel) -- (\xelabel,\yaelabel);
\node[above right] at (\mepsx0,\epsyadlabel) {$\mathsf{variableGadget}_2$};
\draw (\xilabel,\yailabel) -- (\xglabel,\yailabel) node[pos=0.5, above] {$x_3 = \true$};
\draw (\xilabel,\ybdlabel) -- (\xglabel,\ybdlabel) node[pos=0.5, above] {$x_3 = \false$};
\draw (\xglabel,\yailabel) -- (15,\yailabel);
\draw (\xglabel,\ybdlabel) -- (15,\ybdlabel);
\filldraw [fill=lime!30, draw=black] (\mepsx0,\mepsyailabel) rectangle (\epsxilabel, \epsybdlabel);
\draw (\x0,\yailabel) -- (\xglabel,\yailabel);
\draw (\x0,\yailabel) -- (\x0,\ybelabel);
\draw (\xhlabel,\yailabel) -- (\xhlabel,\ybdlabel);
\draw (\xhlabel,\ybdlabel) -- (\xglabel,\ybdlabel);
\draw (\x0,\ybelabel) -- (\xhlabel,\ybelabel);
\node[above right] at (\mepsx0,\epsybdlabel) {$\mathsf{variableGadget}_3$};

	\draw[dashed] (0, -0.4) -- (0, \yglabel) node[above] {$-3L$};
	\draw[dashed] (3, -0.4) -- (3, \yglabel) node[above] {$-2L$};
	\draw[dashed] (6, -0.4) -- (6, \yglabel) node[above] {$-L$};
	\draw[dashed] (9, -0.4) -- (9, \yglabel) node[above] {0};

\end{tikzpicture}
\caption{Layout of the construction depicting gadgets and their interaction.}
\label{fig:segment_apx_whole}
\end{figure}


%% file: main.bbl
\begin{thebibliography}{10}

\bibitem{shrinking_original}
Anna Adamaszek, Parinya Chalermsook, and Andreas Wiese.
\newblock How to tame rectangles: {S}olving {I}ndependent {S}et and {C}oloring
  of rectangles via shrinking.
\newblock In {\em 18th International Conference on Approximation Algorithms for
  Combinatorial Optimization Problems, {APPROX} 2015}, volume~40 of {\em
  LIPIcs}, pages 43--60. Schloss Dagstuhl --- Leibniz-Zentrum f{\"{u}}r
  Informatik, 2015.
\newblock \href {https://doi.org/10.4230/LIPIcs.APPROX-RANDOM.2015.43}
  {\path{doi:10.4230/LIPIcs.APPROX-RANDOM.2015.43}}.

\bibitem{AdamaszekHW19}
Anna Adamaszek, Sariel Har{-}Peled, and Andreas Wiese.
\newblock Approximation schemes for independent set and sparse subsets of
  polygons.
\newblock {\em J. {ACM}}, 66(4):29:1--29:40, 2019.
\newblock \href {https://doi.org/10.1145/3326122} {\path{doi:10.1145/3326122}}.

\bibitem{AroraLMSS98}
Sanjeev Arora, Carsten Lund, Rajeev Motwani, Madhu Sudan, and Mario Szegedy.
\newblock Proof verification and the hardness of approximation problems.
\newblock {\em J. {ACM}}, 45(3):501--555, 1998.
\newblock \href {https://doi.org/10.1145/278298.278306}
  {\path{doi:10.1145/278298.278306}}.

\bibitem{rectangles_apx_hard}
Timothy~M. Chan and Elyot Grant.
\newblock Exact algorithms and {APX}-hardness results for geometric packing and
  covering problems.
\newblock {\em Comput. Geom.}, 47(2):112--124, 2014.
\newblock \href {https://doi.org/10.1016/j.comgeo.2012.04.001}
  {\path{doi:10.1016/j.comgeo.2012.04.001}}.

\bibitem{platypus_book}
Marek Cygan, Fedor~V. Fomin, Lukasz Kowalik, Daniel Lokshtanov, D{\'{a}}niel
  Marx, Marcin Pilipczuk, Micha\l{} Pilipczuk, and Saket Saurabh.
\newblock {\em Parameterized Algorithms}.
\newblock Springer, 2015.
\newblock \href {https://doi.org/10.1007/978-3-319-21275-3}
  {\path{doi:10.1007/978-3-319-21275-3}}.

\bibitem{set_cover_inapproximation}
Irit Dinur and David Steurer.
\newblock Analytical approach to parallel repetition.
\newblock In {\em 46th ACM Symposium on Theory of Computing, {STOC} 2014},
  pages 624--633. {ACM}, 2014.
\newblock \href {https://doi.org/10.1145/2591796.2591884}
  {\path{doi:10.1145/2591796.2591884}}.

\bibitem{ErlebachL10}
Thomas Erlebach and Erik~Jan van Leeuwen.
\newblock {PTAS} for {W}eighted {S}et {C}over on unit squares.
\newblock In {\em 13th International Workshop on Approximation and
  Combinatorial Optimization, {APPROX} 2010}, volume 6302 of {\em Lecture Notes
  in Computer Science}, pages 166--177. Springer, 2010.
\newblock \href {https://doi.org/10.1007/978-3-642-15369-3\_13}
  {\path{doi:10.1007/978-3-642-15369-3\_13}}.

\bibitem{harpeled12}
Sariel Har{-}Peled and Mira Lee.
\newblock Weighted geometric set cover problems revisited.
\newblock {\em J. Comput. Geom.}, 3(1):65--85, 2012.
\newblock \href {https://doi.org/10.20382/jocg.v3i1a4}
  {\path{doi:10.20382/jocg.v3i1a4}}.

\bibitem{hastad}
Johan H\r{a}stad.
\newblock Some optimal inapproximability results.
\newblock {\em J. ACM}, 48(4):798–859, July 2001.
\newblock \href {https://doi.org/10.1145/502090.502098}
  {\path{doi:10.1145/502090.502098}}.

\bibitem{KaraK06}
Jan K{\'{a}}ra and Jan Kratochv{\'{\i}}l.
\newblock Fixed parameter tractability of independent set in segment
  intersection graphs.
\newblock In {\em Second International Workshop on Parameterized and Exact
  Computation, {IWPEC} 2006}, volume 4169 of {\em Lecture Notes in Computer
  Science}, pages 166--174. Springer, 2006.
\newblock \href {https://doi.org/10.1007/11847250\_15}
  {\path{doi:10.1007/11847250\_15}}.

\bibitem{KratschPR16}
Stefan Kratsch, Geevarghese Philip, and Saurabh Ray.
\newblock Point {L}ine {C}over: {T}he easy kernel is essentially tight.
\newblock {\em {ACM} Trans. Algorithms}, 12(3):40:1--40:16, 2016.
\newblock \href {https://doi.org/10.1145/2832912} {\path{doi:10.1145/2832912}}.

\bibitem{LangermanM05}
Stefan Langerman and Pat Morin.
\newblock Covering things with things.
\newblock {\em Discret. Comput. Geom.}, 33(4):717--729, 2005.
\newblock \href {https://doi.org/10.1007/s00454-004-1108-4}
  {\path{doi:10.1007/s00454-004-1108-4}}.

\bibitem{marx05}
D{\'a}niel Marx.
\newblock Efficient approximation schemes for geometric problems?
\newblock In {\em 13th European Symposium on Algorithms, ESA 2005}, pages
  448--459. Springer, 2005.

\bibitem{Marx06}
D{\'{a}}niel Marx.
\newblock Parameterized complexity of independence and domination on geometric
  graphs.
\newblock In {\em 2nd International Workshop on Parameterized and Exact
  Computation, {IWPEC} 2006}, volume 4169 of {\em Lecture Notes in Computer
  Science}, pages 154--165. Springer, 2006.
\newblock \href {https://doi.org/10.1007/11847250\_14}
  {\path{doi:10.1007/11847250\_14}}.

\bibitem{Marx10beat}
D{\'{a}}niel Marx.
\newblock Can you beat treewidth?
\newblock {\em Theory Comput.}, 6(1):85--112, 2010.
\newblock \href {https://doi.org/10.4086/toc.2010.v006a005}
  {\path{doi:10.4086/toc.2010.v006a005}}.

\bibitem{voronoi}
D{\'{a}}niel Marx and Micha\l{} Pilipczuk.
\newblock Optimal parameterized algorithms for planar facility location
  problems using voronoi diagrams.
\newblock {\em {ACM} Trans. Algorithms}, 18(2):13:1--13:64, 2022.
\newblock \href {https://doi.org/10.1145/3483425} {\path{doi:10.1145/3483425}}.

\bibitem{settling_apx_hardness}
Nabil~H. Mustafa, Rajiv Raman, and Saurabh Ray.
\newblock Settling the apx-hardness status for geometric set cover.
\newblock In {\em 55th {IEEE} Annual Symposium on Foundations of Computer
  Science, {FOCS} 2014}, pages 541--550. {IEEE} Computer Society, 2014.
\newblock \href {https://doi.org/10.1109/FOCS.2014.64}
  {\path{doi:10.1109/FOCS.2014.64}}.

\bibitem{unit_disks}
Nabil~H. Mustafa and Saurabh Ray.
\newblock Improved results on geometric hitting set problems.
\newblock {\em Discret. Comput. Geom.}, 44(4):883--895, 2010.
\newblock \href {https://doi.org/10.1007/s00454-010-9285-9}
  {\path{doi:10.1007/s00454-010-9285-9}}.

\bibitem{shrinking1}
Micha\l{} Pilipczuk, Erik~Jan van Leeuwen, and Andreas Wiese.
\newblock Approximation and parameterized algorithms for geometric independent
  set with shrinking.
\newblock In {\em 42nd International Symposium on Mathematical Foundations of
  Computer Science, {MFCS} 2017}, volume~83 of {\em LIPIcs}, pages 42:1--42:13.
  Schloss Dagstuhl --- Leibniz-Zentrum f{\"{u}}r Informatik, 2017.
\newblock \href {https://doi.org/10.4230/LIPIcs.MFCS.2017.42}
  {\path{doi:10.4230/LIPIcs.MFCS.2017.42}}.

\bibitem{voronoi_true}
Micha\l{} Pilipczuk, Erik~Jan van Leeuwen, and Andreas Wiese.
\newblock Quasi-polynomial time approximation schemes for packing and covering
  problems in planar graphs.
\newblock {\em Algorithmica}, 82(6):1703--1739, 2020.
\newblock \href {https://doi.org/10.1007/s00453-019-00670-w}
  {\path{doi:10.1007/s00453-019-00670-w}}.

\bibitem{shrinking2}
Andreas Wiese.
\newblock Independent {S}et of convex polygons: From $n^\epsilon$ to
  $1+\epsilon$ via shrinking.
\newblock {\em Algorithmica}, 80(3):918--934, 2018.
\newblock \href {https://doi.org/10.1007/s00453-017-0347-8}
  {\path{doi:10.1007/s00453-017-0347-8}}.

\end{thebibliography}
